%% file: 10978-test.tex
\documentclass{fundam}

\usepackage[english]{babel}

\usepackage{amssymb}
\usepackage{mathrsfs}
\usepackage{url}
\usepackage{tikz}
\usetikzlibrary{arrows,shapes,decorations,automata,backgrounds,petri,calc,positioning}
\usepackage{enumerate}
\usepackage{xspace}
\usepackage{stmaryrd}
\usepackage[linesnumbered,ruled,lined]{algorithm2e}

\usepackage{macros}

\begin{document}

\setcounter{page}{179}
\publyear{22}
\papernumber{2146}
\volume{188}
\issue{3}

  \finalVersionForARXIV
  %%  \finalVersionForIOS

\title{Structural Liveness of Immediate Observation Petri Nets}

\author{Petr Jan\v{c}ar and Ji\v{r}{\'{i}} Val\r{u}\v{s}ek\thanks{Supported by Grant No.\ IGA\_PrF\_2022\_018 and
           IGA\_PrF\_2021\_022 of IGA of  Palack\'y University Olomouc.}\thanks{Address for correspondence:
           Dept of Computer Science, Faculty of Science, Palack\'y University in Olomouc, Czech Republic.  \newline \newline
                    \vspace*{-6mm}{\scriptsize{Received November 2022; \ accepted February  2023.}}}
\\
Dept of Computer Science, Faculty of Science\\
Palack\'y University in Olomouc\\
Czech Republic\\
\{petr.jancar, jiri.valusek01\}{@}upol.cz}
\maketitle

\runninghead{P. Jan\v{c}ar and J. Val\r{u}\v{s}ek}{Structural Liveness of Immediate Observation Petri Nets}

\maketitle

\begin{abstract}
	We look in detail at the structural liveness problem (SLP) for subclasses
	of Petri nets, namely immediate observation nets (IO nets)
and their generalized variant called branching immediate multi-observation nets (BIMO nets), that were
recently introduced by Esparza, Raskin, and Weil-Kennedy.
We show that SLP is PSPACE-hard for IO nets and in PSPACE for BIMO
	nets. In particular, we discuss the (small) bounds on the
	token numbers in net places that are
	decisive for a marking to be (non)live.

	\textbf{2012 ACM Subject Classification:} Theory of computation
$\rightarrow$ Logic and verification

\end{abstract}


\section{Introduction}

Petri nets are an established model of concurrent systems,
and a natural part of related research
aims to clarify computational complexity of verifying basic
behavioural properties for various (sub)classes of this model.
A famous example is the reachability problem for standard
place/transition Petri nets, which was recently shown to be
Ackermann-complete~(\cite{DBLP:conf/lics/LerouxS19,DBLP:journals/corr/abs-2104-12695,
DBLP:journals/corr/abs-2104-13866}).

Here we are interested in the \emph{structural liveness problem} (SLP): given
a net $N$, is there an initial marking $M_0$ such that the marked net
$(N,M_0)$ is live? We recall that a marked net $(N,M_0)$ is live if no
transition can become disabled forever, in markings reachable from
$M_0$.
The \emph{liveness problem} (LP), asking if a marked net $(N,M_0)$ is live,
has been long
known to be tightly related to the reachability
problem~\cite{DBLP:phd/ndltd/Hack76};
hence LP has turned out to be Ackermann-complete as well.
Somewhat surprisingly, for SLP even
the decidability status was open until recently~(see,
e.g.,~\cite{DBLP:journals/ipl/BestE16}); the currently known status is
that the problem is
EXPSPACE-hard and decidable~\cite{DBLP:journals/acta/JancarP19}.

Here we look at the complexity of SLP for a subclass of place/transition Petri nets,
namely for the class of \emph{immediate observation Petri nets} (IO
nets), and their generalized variant called
 \emph{branching immediate multi-observation Petri nets} (BIMO nets).
These models
were introduced and studied recently,
in~\cite{DBLP:conf/apn/EsparzaRW19,DBLP:conf/rp/RaskinW20,DBLP:conf/concur/RaskinWE20},
motivated by a study of population protocols
and chemical reaction networks.
\emph{Population protocols}~\cite{DBLP:journals/dc/AngluinADFP06}
are a model of computation where an arbitrary number of
indistinguishable finite-state agents interact in pairs; an
interaction of two agents, being in states $q_1$ and $q_2$, consists
in changing their states to $q_3$ and $q_4$, respectively,
according to a transition function. A~global state of a~protocol
is just a function assigning to each (local) state the number of agents in
this state.
It is natural to represent
population protocols by
Petri nets where places represent (local) states, and markings
represent global states.
The above mentioned IO nets represent a special subclass,
so called  \emph{immediate observation
protocols} (IO protocols), that was introduced
in~\cite{DBLP:journals/dc/AngluinAER07}. Here an agent can change
its state $q_1$ to $q_2$ when observing that another agent is
in state $q_3$.
We also note that BIMO nets can be viewed as a generalization of Petri
nets related to basic parallel processes (BPP nets), studied, e.g.,
in~\cite{DBLP:journals/fuin/Esparza97}. In the BPP nets each
transition $t$ has precisely one input place $p$, the edge $(p,t)$
having the weight $1$. In the BIMO nets, for performing such a
transition $t$ it is not only necessary that the place $p$ has at
least one token but there is also a context-condition, requiring that
some places have sufficient amounts of tokens.
(The relation of BPP nets and BIMO nets resembles the relation of
context-free and context-sensitive grammars, where the
word-concatenation is viewed as commutative.)

Among the results of~\cite{DBLP:conf/apn/EsparzaRW19} is
the PSPACE-completeness of the liveness problem (LP) for IO
nets. The paper~\cite{DBLP:conf/apn/EsparzaRW19}
does not deal with the structural liveness problem
(SLP) directly but it can be derived from its results that SLP is in PSPACE for IO
nets.

\emph{Our contribution.}
We first show that also SLP is PSPACE-hard for IO nets.
Here we proceed similarly as~\cite{DBLP:conf/apn/EsparzaRW19} where the
hardness proof for LP is given; we show a modification of a~standard simulation
of linear bounded automata by $1$-safe nets,
but the construction for SLP is more subtle than for LP.
The PSPACE membership is straightforward for LP on IO nets,
since IO nets are a special case of conservative nets, but it is
not so straightforward for SLP.
We show that for a BIMO net, where $P$ is the set of places and $w$
the maximum edge-weight, the fact whether or not a marking $M$ is live
($M:P\rightarrow\Nat$ assigns the number of tokens to each place)
is determined by the values $M(p)$ that are less than $2\cdot
w\cdot|P|$.
This result allows us to give a simple proof that SLP is in PSPACE (and thus
PSPACE-complete) for BIMO nets.

\emph{The organization of the paper.} In
Section~\ref{sec:basicdef} we give the basic definitions related to
Petri nets, and define the subclasses (BIMO, BIO, IMO, IO nets) in
which we are interested; in part~\ref{subsec:results}
we summarize our results. Section~\ref{sec:PSPACEhardness} shows
the \PSPACE-hardness of the structural liveness problem (SLP) for
ordinary IO nets (where ``ordinary'' means that all edge-weights are
$1$). In Sections~\ref{sec:strategy}, \ref{sec:cruclemma} and
\ref{sec:upperord} we prove the announced results for ordinary BIMO, BIO,
IMO and IO nets. Section~\ref{sec:extensions} extends these results to all BIMO,
BIO, IMO, IO nets, by a simple construction. In
Section~\ref{sec:inpspace} we use the achieved results to show
that SLP for BIMO nets is in PSPACE.
Some additional remarks are given in
Section~\ref{sec:addrem}.

\section{Preliminaries, and results}\label{sec:basicdef}

By $\Nat$ we denote the set $\{0,1,2,\dots\}$ of nonnegative
integers, and we put $[i,j]=\{i,i{+}1,\dots,j\}$ for
$i,j\in\Nat$.\vspace*{-1mm}

\paragraph{Multisets.}
Given a set $U$, called the \emph{universe},
by a \emph{multiset} $M$ \emph{over} $U$ we mean a function $M:U
\rightarrow \Nat$; for $x\in U$ we write $x\in M$ if $M(x)\geq 1$.
We use the notation $M = \Lbag
x_1,x_2,\dots,x_n \Rbag$ for finite multisets; here we have
$M(x) = \sizeof{\{i\in[1,n]\,; x_i=x\}}$,
and we put $\sizeof{M}=n$.
 A set $X\subseteq U$ is
naturally viewed as
a multiset $X:U\rightarrow\{0,1\}$.
By $\emptyset$ we denote the empty (multi)set ($\emptyset(x)=0$
for all $x\in U$).

\medskip
Given two multisets $M,M'$ over $U$, we define
the multisets $M''=M+M'$
and  $M'''=M-M'$
so that $M''(x)=M(x)+M'(x)$
and $M'''(x)=\max\{M(x)-M'(x),0\}$ for all $x\in U$.
By $M\leq M'$
we denote that $M(x)\leq M'(x)$ for all $x\in U$.
We also use the intersection of multisets:
for $M''=M\cap M'$ we have $M''(x)=\min\{M(x),M'(x)\}$ for all
$x\in U$.

\subsection{Standard Petri net definitions}

\paragraph{Nets, subnets (ordinary, and weighted).}
A \emph{net} $N$ is a triple $(P,T,F)$ where $P$ and $T$ are
finite disjoint sets of \emph{places} and
\emph{transitions}, respectively, and $F: (P \times T) \cup
(T \times P) \longrightarrow \mathbb{N}$ is a \emph{flow
function}. A pair $(x,y)\in (P \times T) \cup (T \times P)$
where $F(x,y)\geq 1$ is also called an \emph{edge} in $N$,
and $F(x,y)$ is viewed as its \emph{weight}.
A \emph{net} $N=(P,T,F)$ is \emph{ordinary} if $F$ is of the
type $(P \times T) \cup (T \times P) \longrightarrow
\{0,1\}$ (hence the weights of edges are $1$).

\medskip
We use a standard depiction of nets; for instance, the net
in Figure~\ref{fig:ex1} has $6$ places (circles), $6$ transitions
(boxes), and
the depicted edges; the edge weights $1$ are implicit. If
the weight is larger than $1$, then it is depicted
explicitly, like, e.g., the weight $3$
in Figure~\ref{fig:bimoext} (in Section~\ref{sec:extensions}).

\begin{figure}[!h]
    \centering
    \input{figures/example-basic-definitions.tex}
    \caption{Example of a marked \oBIMO net.}
    \label{fig:ex1}
\end{figure}
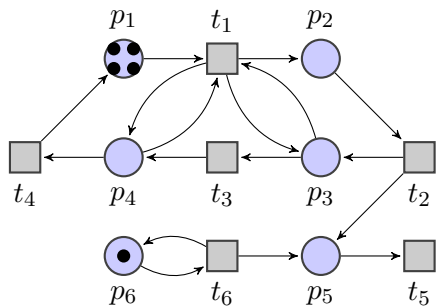

\medskip
Given a net  $N = (P,T,F)$ and sets  $P' \subseteq P$, $T'
\subseteq T$, by $\rst{N}{(P',T')}$ we denote the (sub)net
$(P',T',F')$ where $F'$ arises from $F$ by restricting its
domain to
$(P' \times T') \cup (T' \times P')$.
Sometimes we also deal with subnets arising by removing some edges. \vspace*{-1mm}

\paragraph{Pre-msets, post-msets, siphons.}
Let $N = (P,T,F)$ be a fixed net. For each transition $t\in T$ we
define its \emph{pre-mset} $\premset{t}$
and its \emph{post-mset} $\postmset{t}$
as
multisets over $P$ where $\premset{t}(p) = F(p,t)$
and $\postmset{t}(p) = F(t,p)$, for each place $p\in P$.

  A set of places $S\subseteq P$ is a \emph{siphon} if
for each $t\in T$ we have that $\postmset{t}\cap S\neq\emptyset$
entails  $\premset{t}\cap S\neq\emptyset$.
For instance, the
    set $S=\{p_2,p_3,p_4\}$ in  Figure~\ref{fig:ex1} is a
    siphon, since $\{t\in T\mid \postmset{t}\cap S\neq\emptyset\}=
\{t_1,t_2,t_3\}$ and
    $\{t\in T\mid \premset{t}\cap S\neq\emptyset\}=\{t_1,t_2,t_3,t_4\}$. \vspace*{-1mm}

\paragraph{Markings, marked nets.}
Given  a net  $N = (P,T,F)$, a \emph{marking} $M$ of $N$ is
a multiset over $P$ (hence $M: P \rightarrow \mathbb{N}$),
where $M(p)$
is viewed as the number of \emph{tokens} on the place $p$
(alternatively we also say ``in the place $p$''). For
$P'\subseteq P$, $\rst{M}{P'}$ denotes the restriction of
$M$ to $P'$.
A \emph{place} $p\in P$
is \emph{marked at} $M$ if $M(p)\geq 1$;
a~\emph{set of places} $P'\subseteq P$
is \emph{marked at} $M$ if
$|\rst{M}{P'}|\geq 1$.

When the places in $P$ are ordered, we can also naturally
 view markings as vectors; e.g., the marking $M$ depicted in
 Figure~\ref{fig:ex1} can be given as $(4,0,0,0,0,1)$. By
 $\bzero$ we denote the zero vector (with the dimension
 clear from context).
E.g., for the siphon $S=\{p_2,p_3,p_4\}$ in
    Figure~\ref{fig:ex1} we have $\rst{M}{S}=\bzero$, i.e., the
    siphon $S$ is unmarked at $M$.

A \emph{marked net} (or a \emph{Petri net}) is a tuple
$(N,M_0)$ where $N$ is a net and $M_0$ is a marking of $N$,
called the \emph{initial marking}. \vspace*{-1mm}

\paragraph{Executions, reachability.} Given  a net  $N =
(P,T,F)$, a \emph{transition} $t$ is \emph{enabled at a
marking} $M$, which is denoted by $M\step{t}$,
if $M\geq \premset{t}$
(i.e., $M(p) \geq F(p,t)$ for all $p\in P$).
 If $t$ is enabled at $M$,
 it can \emph{fire} (or \emph{be performed}, or \emph{be
executed}), which yields the marking $M'=(M-\premset{t})+\postmset{t}$
(hence $M'(p)= M(p) - F(p,t) + F(t,p)$ for all $p\in P$); this is denoted
by $M \step{t}M'$.

A sequence  $M_0 \step{t_1} M_1 \step{t_2} M_2\cdots
 \step{t_k} M_k$ is called an \emph{execution}, \emph{from} $M_0$
\emph{to} $M_k$, which is also presented as  $M_0 \step{\sigma}
 M_k$ where $\sigma=t_1t_2\cdots t_k$. A marking $M'$ is
 \emph{reachable from} $M$ if there is an execution $M
 \step{\sigma} M'$. By  $\rset{M}$ we denote the set of all
 markings that are reachable from $M$;
we also write $M\step{*}M'$ instead of $M'\in\rset{M}$.

For instance, an execution of the net in
Figure~\ref{fig:ex1} is
$(1,1,1,1,1,1)\step{t_2}(1,0,2,1,2,1)\step{t_3}(1,0,1,2$, $2,1)\step{t_3}(1,0,0,3,2,1)
\step{t_4}(2,0,0,2,2,1)\step{t_4}(3,0,0,1,2,1)\step{t_4}(4,0,0,0,2,1)
\step{t_5}(4,0,0,0$, $1,1)\step{t_5}(4,0,0,0,0,1)$.
We might also note that generally any unmarked siphon $S$
cannot get marked; i.e., $\rst{M}{S}=\bzero$ entails
$\rst{M'}{S}=\bzero$ for all $M'\in\rset{M}$. \vspace*{-1mm}

\paragraph{Dead and live transitions, liveness and structural liveness.}
 Given $N = (P,T,F)$,
a \emph{transition} $t$ is \emph{dead at} a \emph{marking}
$M$ if there is no $M'\in\rset{M}$ such that $M'\step{t}$
(hence $t$ is disabled in all markings reachable from $M$).
A~\emph{transition} $t$ is \emph{live at} $M$ if it is
non-dead at each $M'\in \rset{M}$. We note that a~transition can
be both non-live and non-dead at $M$.

A \emph{marking} $M$ of $N$ is \emph{live} if all
transitions are live at $M$.
A \emph{marked net} $(N,M_0)$ is \emph{live} if
$M_0$ is live (for $N$).
A net $N$ is
\emph{structurally live} if there is $M_0$ such that
$(N,M_0)$ is live.

For instance, the net in  Figure~\ref{fig:ex1} is clearly
not structurally live (if $M(p_6)=0$, then $M$ is clearly
non-live, and otherwise there is $M'\in\rset{M}$ such that
the siphon $S=\{p_2,p_3,p_4\}$ is unmarked at $M'$, i.e.\
$\rst{M'}{S}=\bzero$). If we removed the edges
$(p_3,t_1),(t_1,p_3)$ and $(p_4,t_1),(t_1,p_4)$, then the
net would become structurally live.

The \emph{liveness problem} (LP) asks if a given marked net
$(N,M_0)$ is live. The \emph{structural liveness problem}
(SLP) asks if a given net $N$ has a marking $M_0$ for which
$(N,M_0)$ is live. \vspace*{-1mm}

\paragraph{Conservative nets, and the (structural) liveness problem.}
We call a \emph{net} $N=(P,T,F)$ \emph{conservative}
if $\sizeof{\premset{t}}=\sizeof{\postmset{t}}$ for each $t\in T$;
in this case $M\step{t}M'$ entails $\sizeof{M}=\sizeof{M'}$
(hence in every execution the
number of tokens is constant). We remark that the
definition of conservative nets in the literature is sometimes more general,
in which case our notion
corresponds to $\mone$-conservative nets.

For the results of this paper, it is particularly useful
to recall the well-known fact:

\begin{proposition}\label{p21}
	The liveness problem (LP) for conservative nets
	is \PSPACE-complete.
\end{proposition}
The \PSPACE-hardness follows by a standard reduction from the acceptance
 problem for linear bounded automata (LBA),
even for ordinary conservative nets,
 as we also recall later in detail.
For the membership in \PSPACE we can refer
to~\cite{DBLP:journals/tcs/JonesLL77}; this holds even for the case
with general edge-weights that are given in binary.

\begin{remark}\label{rem:pspace}
It is useful to recall the idea of the \PSPACE-membership:
Given a~conservative net $N$, a transition $t$ and a marking $M$,
deciding if $t$ is non-dead
at $M$ is obviously in NPSPACE (we just perform
a nondeterministically chosen execution from $M$ until covering
$\premset{t}$, i.e., until reaching $M'$ such that
$M'\geq\premset{t}$). Since NPSPACE$=$PSPACE, we deduce that deciding if $t$ is dead
at $M$ is in PSPACE.
Given a conservative net $N$ and a~marking $M$, deciding if
there is $M'\in\rset{M}$ and a transition $t$ that is dead at $M'$ is
thus in NPSPACE, hence in PSPACE, as well.
\end{remark}

 \begin{remark}
 The liveness problem (LP) for general nets is well-known to be
 tightly related to the reachability problem; the recent break-through
 results~\cite{DBLP:journals/corr/abs-2104-13866,DBLP:journals/corr/abs-2104-12695,DBLP:conf/lics/LerouxS19} thus show its huge computational complexity,
 namely the Ackermann-completeness.
The structural liveness problem (SLP) is more unclear so far. For
general nets SLP is known to be \EXPSPACE-hard and
decidable~\cite{DBLP:journals/acta/JancarP19}; for conservative nets
we can show that SLP is \EXPSPACE-hard and elementary~\cite{JLV22}.
 \end{remark}

\subsection{Immediate observation nets, and their (more general) variants}

We recall the notions of
IO (immediate observation) nets and BIO (branching IO) nets, including
 the multi-observer versions: IMO and BIMO nets.
These nets were
introduced in
\cite{DBLP:conf/apn/EsparzaRW19,DBLP:conf/concur/RaskinWE20},
being originally motivated by
 (special types of) population protocols.
They have restricted types of transitions; we start with defining the
most general case.\vspace*{-1mm}

\paragraph{Branching immediate multiple-observation (BIMO) transitions.}
Given a net $N=(P,T,F)$, we say that a transition $t\in T$ is
a~\emph{BIMO transition} if
$\sizeof{\premset{t}-\postmset{t}} \leq 1$
(hence there is at most one place $p$ for which
$\premset{t}(p)>\postmset{t}(p)$, in which case we have
$\premset{t}(p)-\postmset{t}(p)=1$).

\medskip
\noindent
\emph{Convention.}
In our considerations, for each BIMO transition $t$ we will assume that
$\premset{t}\neq\emptyset$. I.e., in the case $\premset{t}=\emptyset$
we tacitly assume an additional ``dummy place'' $\textsc{d}$ such that
 $\premset{t}(\textsc{d})=\postmset{t}(\textsc{d})=1$ and $\textsc{d}$
 is marked in all considered markings.

\medskip

Having the convention in mind,
to each BIMO transition $t$ we fix a presentation
\begin{center}
    $\bimotrans{t}{p_s}{
	    p_{o_1},p_{o_2},\dots,p_{o_\ell}}{p_{d_1},p_{d_2},\dots,p_{d_k}}$
\end{center}
(for some $k,\ell\in\Nat$)
where
\begin{center}
$\premset{t}=\Lbag p_s\Rbag +\Lbag p_{o_1},p_{o_2},\dots,p_{o_\ell}
		\Rbag$, and
$\postmset{t}=\Lbag p_{o_1},p_{o_2},\dots,p_{o_\ell}\Rbag +
\Lbag p_{d_1},p_{d_2},\dots,p_{d_k}  \Rbag$.
\end{center}
We note that the multisets $\Lbag p_s\Rbag$,
$\Lbag p_{o_1},p_{o_2},\dots,p_{o_\ell}\Rbag$, and
$\Lbag p_{d_1},p_{d_2},\dots,p_{d_k}  \Rbag$ are not necessarily
disjoint (their pairwise intersections might be nonempty).

\medskip
The place $p_s$ is called the \emph{source place of} $t$. If
$\sizeof{\premset{t}-\postmset{t}} = 1$, then $p_s$ is the unique
place satisfying $\premset{t}(p_s)=1+\postmset{t}(p_s)$;
if $\sizeof{\premset{t}-\postmset{t}} = 0$ (hence
$\premset{t}(p)\leq\postmset{t}(p)$ for all $p\in P$), then we fix
$p_s$ as one of the places $p$ for which $\premset{t}(p)\geq 1$.

\medskip
For instance, all transitions in Figure~\ref{fig:bimoext}
(in Section~\ref{sec:extensions}) are BIMO-transitions; the
only presentation of $t$ in Figure~\ref{fig:bimoext}(left)
is $\bimotrans{t}{p_1}{p_1}{p_1,p_1,p_2}$, $t'$ in
Figure~\ref{fig:bimoext}(right) can be presented as
$\bimotrans{t'}{p_\lara{1,1}}{p_\lara{1,2}}{p_\lara{1,1},p_\lara{1,3},p_\lara{2,1}}$
or as
$\bimotrans{t'}{p_\lara{1,2}}{p_\lara{1,1}}{p_\lara{1,2},p_\lara{1,3},p_\lara{2,1}}$,
and $t_\lara{1,1}$ only as
$\biotrans{t_\lara{1,1}}{p_\lara{1,1}}{\emptyset}{p_\lara{1,2}}$
(which is also written as
$\biotrans{t_\lara{1,1}}{p_\lara{1,1}}{}{p_\lara{1,2}}$).

\medskip
We observe that performing a BIMO-transition
  $\bimotrans{t}{p_s}{
	    p_{o_1},p_{o_2},\dots,p_{o_\ell}}{p_{d_1},p_{d_2},\dots,p_{d_k}}$
(i.e., executing a~step $M\step{t}M'$)
can be viewed so that a ``source'' token from $p_s$ has ``branched'' into new
tokens in the \emph{destination places} constituting the set
$\{p_{d_1},p_{d_2},\dots,p_{d_k}\}$; the new tokens are created
in the destination places
with
the multiplicities determined by the multiset $\Lbag
p_{d_1},p_{d_2},\dots,p_{d_k}\Rbag$. We note that it is not excluded
that $k=0$ (in which case the source token disappears since there are no destination
places) or that
$p_{d_i}=p_s$
for some $i$ (which is the case of the transition
on the left of Figure~\ref{fig:bimoext}). Performing $t$
is conditioned not only on
the presence of a~token in the source place $p_s$ but also
on the presence of enough tokens in the \emph{observation
places} constituting the set
$\{p_{o_1},p_{o_2},\dots,p_{o_\ell}\}$; this ``enough
tokens'' is determined by the multiset $\Lbag
p_{o_1},p_{o_2},\dots,p_{o_\ell}\Rbag$.\vspace*{-1mm}

\paragraph{BIMO transitions of the type BIO, IMO, IO.}
Given a BIMO transition
\begin{center}
    $\bimotrans{t}{p_s}{
	    p_{o_1},p_{o_2},\dots,p_{o_\ell}}{p_{d_1},p_{d_2},\dots,p_{d_k}}$
\end{center}
we say that $t$ is:
\begin{itemize}
	\item		
a \emph{BIO transition}
		if the multiset $\Lbag p_{o_1},p_{o_2},\dots,p_{o_\ell}\Rbag$
		is a singleton set $\{p_o\}$ or the empty set
		(performing $t$ is
		conditioned on at most
		one observation-token, which is the case
		for all transitions
		in Figure~\ref{fig:bioce});
 in this case we also write
		\begin{center}
			$\biotrans{t}{p_s}{p_{o}}{p_{d_1},p_{d_2},\dots,p_{d_k}}$
			\ or\ \ $\biotrans{t}{p_s}{}{p_{d_1},p_{d_2},\dots,p_{d_k}}$;
		\end{center}			
	\item		
		an \emph{IMO transition}
		if
		the multiset $\Lbag
		p_{d_1},p_{d_2},\dots,p_{d_k}\Rbag$ is a singleton set
		$\{p_d\}$
(the source token does not
		branch, nor disappears, it just moves from $p_s$ to $p_d$);
in this case we also write
		\begin{center}
		$\imotrans{t}{p_s}{p_{o_1},p_{o_2},\dots,p_{o_\ell}}{p_{d}}$;
		\end{center}
	\item		
		an \emph{IO transition} if
it is a BIO and IMO transition;
in this case we also write
		\begin{center}
		$\iotrans{t}{p_s}{p_{o}}{p_{d}}$ or $\iotrans{t}{p_s}{}{p_{d}}$.
		\end{center} \vspace*{-1mm}
\end{itemize}

\paragraph{Nets of the type BIMO, ord-BIMO, BIO, ord-BIO, IMO,
ord-IMO, IO, ord-IO.}
For a type $\textrm{X}\in\{\textrm{BIMO, BIO, IMO, IO}\}$ we say that a net $N$ is an
\emph{X net} if all its transitions are X transitions; moreover,
$N$ is an \emph{ord-X net} if $N$ is an X net
that is
ordinary (the edge weights are just $1$).

We note that for any BIMO transition  $\bimotrans{t}{p_s}{
	    p_{o_1},p_{o_2},\dots,p_{o_\ell}}{p_{d_1},p_{d_2},\dots,p_{d_k}}$
 in an ordinary net it holds
that $\Lbag p_{o_1},p_{o_2},\dots,p_{o_\ell}\Rbag$
and $\Lbag p_{d_1},p_{d_2},\dots,p_{d_k}\Rbag$ are two disjoint
\emph{sets}, and $p_s\not\in\{p_{o_1},p_{o_2},\dots,p_{o_\ell}\}$
(while we still can have  $p_s\in\{p_{d_1},p_{d_2},\dots,p_{d_k}\}$).

\medskip

We observe that IMO nets are conservative, hence
Proposition~\ref{p21} entails:

\begin{proposition}\label{p21corol}
	The liveness problem (LP) for IMO nets
	is in \PSPACE.
\end{proposition}
We recall that this also holds when the edge-weights
(i.e., the
multiplicities of observation
 places) are given in binary.
 Moreover, the PSPACE-hardness proof for
 conservative nets has been enhanced in~\cite{DBLP:conf/apn/EsparzaRW19}
to show that LP is PSPACE-hard also for IO nets.

 \begin{table}[!b]
 \vspace*{-3mm}
  \caption{Given a structurally live net with maximum edge-weight $w$,
		there is a live marking in which each component is 	bounded by
				the 1st upper bound; moreover, the (non)liveness status of
				any marking does not change if each component greater than
				the 2nd upper bound is replaced with this bound.
			} 	\label{tab:res}\vspace*{-3mm}
				\begin{center}
	\scalebox{0.94}{
		\begin{tabular}{|c|c|c|}
				\hline
			 \textbf{Class of nets} & \textbf{1st upper bound} & \textbf{2nd upper bound} \\
		   \hline
				\oIO and \oIMO & $1$ & $2\cdot\sizeof{P}$ \\
			 IO & $2$ & $4\cdot\sizeof{P}$ \\
			 IMO & $w$ & $2\cdot w\cdot\sizeof{P}$ \\
			 \oBIO and \oBIMO & $\sizeof{P}$ & $2\cdot\sizeof{P}$ \\
			 BIO and BIMO & $w\cdot\sizeof{P}$ & $2\cdot w\cdot \sizeof{P}$ \\
			 \hline
			\end{tabular} }
		\end{center}
		\end{table}

\subsection{Results}\label{subsec:results}
Below we summarize the results of this paper.

\begin{enumerate}[a)]
\item By Theorem~\ref{thm1} and its proof we show that  a modification of the hardness
proof for the liveness problem (LP) for IO nets in~\cite{DBLP:conf/apn/EsparzaRW19} can be enhanced
		to demonstrate the  PSPACE-hardness of the structural liveness problem (SLP) for ord-IO nets.
		(We remark that SLP 		is EXPSPACE-hard for conservative nets~\cite{JLV22}.)

\item Table~\ref{tab:res} summarizes our results concerning the sizes of live
markings in the mentioned net classes, as stated by the following  theorem.
\end{enumerate}

\begin{theorem}\label{thm:upperboundstable}
			Given a structurally live net $N=(P,T,F)$
			of a type in the first
	column of Table~\ref{tab:res}, with the maximum edge-weight
	$w$ (where $w=1$ if $N$ is ordinary), then
	\begin{enumerate}[1.]
		\item
			there is a live marking $M$ of $N$
			in which $M(p)$ is no bigger than
	    the respective 1st upper bound in Table~\ref{tab:res}, for
		each $p\in P$;
    \item
whether or not a marking $M$ of $N$ is live is determined
		by the restriction $\rst{M}{P'}$ where $P'$
		consists of the places $p$ for which
		the values $M(p)$ are
		smaller
than
	    the respective 2nd upper bound in Table~\ref{tab:res}.
\end{enumerate}
	\end{theorem}	
	
For instance, if an ord-IMO net
$N=(P,T,F)$
is
structurally live, then there is $M_0:P\rightarrow\{0,1\}$ such that
$(N,M_0)$ is live; moreover, $(N,M)$ is live iff  $(N,M')$ is live
where $M'(p)=M(p)$ if $M(p)< 2\cdot\sizeof{P}$ and
 $M'(p)=2\cdot\sizeof{P}$ otherwise (for all $p\in P$).

 \begin{enumerate}[c)]
\item By Theorem~\ref{thm:BIMOinPSPACE} we show that the structural
	liveness problem (SLP) for BIMO nets is in
	PSPACE (and thus PSPACE-complete).
\end{enumerate}

Regarding the result c), we note that
our bounds from b) for IMO nets
(already the 1st upper bound, in fact)
immediately show that SLP for IMO nets is
in PSPACE (and thus PSPACE-complete), by recalling
Proposition~\ref{p21corol}.
For general BIMO nets we show the PSPACE-membership of SLP by
using further ideas (one of them being captured by
Lemma~\ref{lem:crucialingred} in particular).

%\medskip
\begin{remark}
The papers~\cite{DBLP:conf/apn/EsparzaRW19,DBLP:conf/concur/RaskinWE20},
and in particular the PhD thesis by Chana Weil-Kennedy in preparation,
contain a detailed study of subclasses of BIMO nets, concentrating
mainly on
the general analysis questions like reachability and coverability.
The \PSPACE-membership of structural liveness could be also derived from the published
results for IO nets, while for BIMO nets this will follow from
the mentioned PhD thesis in preparation.
Nevertheless, our direct handling of structural
liveness for BIMO nets yields stronger bounds captured by Table~\ref{tab:res}
and a more specific insight into this problem (that remains so far a bit
elusive for more general nets).
\end{remark}

\section{\PSPACE-hardness of structural liveness for \oIO nets}\label{sec:PSPACEhardness}
In this section, we show that the
structural liveness problem (SLP) for \oIO nets is
\PSPACE-hard; this is achieved by an enhancement of (a
modification of) the construction showing \PSPACE-hardness
of the liveness problem (LP)
from~\cite{DBLP:conf/apn/EsparzaRW19}.
(This lower bound is later matched by a \PSPACE upper
bound that holds for the most general of the considered classes, i.e., for
BIMO nets, in which the edge-weights are given in binary.)

We first introduce the notion of carriers of markings,
and
an observation captured by
Proposition~\ref{prop:greatermarking}:
Informally speaking, if we add a token onto a place $p$ in a marking $M$
where $M(p)\geq 1$ (in an \oIO net $N$), then this additional token
could be imagined as pasted down to an original token with which it can
then be moving together (by repeating the transitions moving the original
token), whereas the resulting marking carrier remains the same.

\begin{remark}
Later we will look at this observation in more detail for ord-BIMO nets,
but now it suffices to deal with ord-IO nets.
\end{remark}

\paragraph{Carriers of markings.}
Given a net $N=(P,T,F)$ and a marking $M$ of $N$, by
$\carrier{M}$ we denote the \emph{carrier of} $M$, i.e.\ the
set $\{p\in P\mid M(p)\geq 1\}$.

\begin{proposition}[Added tokens can be pasted down to original
	tokens in ord-IO nets]\label{prop:greatermarking}
We assume an \oIO net $N = (P,T,F)$ and its execution $M\step{\sigma}M'$.
Then for any $\bar{M}\geq M$ such that $\carrier{\bar{M}}=\carrier{M}$
	there are $\bar{\sigma}$ and $\bar{M}'\geq M'$ for which
	$\bar{M}\step{\bar{\sigma}}\bar{M}'$ and
$\carrier{\bar{M}'}=\carrier{M'}$.
\end{proposition}

\begin{proof}
	We prove the claim by induction on the value $\length{\sigma}$
	(the length of $\sigma$).
We thus assume that
	$M\step{\sigma}M'$, $\bar{M}\geq M$,
	$\carrier{\bar{M}}=\carrier{M}$, and that the claim holds for all
	$\sigma'$ shorter than $\sigma$ (by the induction hypothesis).
In the case $\length{\sigma}=0$ the claim is trivial, so we now assume
	that $\sigma=t\,\sigma'$, where
	$\iotrans{t}{p_s}{p_{o}}{p_{d}}$ (or $\iotrans{t}{p_s}{}{p_{d}}$)
	and $M\step{t}M''\step{\sigma'}M'$.
We thus have
	$\bar{M}(p_s)\geq M(p_s)\geq 1$, and
	$\bar{M}\step{t^d}\bar{M}''$  where
	$d=1+(\bar{M}(p_s){-}M(p_s))$. It is clear that
 $\bar{M}''\geq M''$ and
	$\carrier{\bar{M}''}=\carrier{M''}$; hence we can apply the
	induction hypothesis to $M''\step{\sigma'}M'$ and $\bar{M}''$,
	which entails that
	$\bar{M}\step{t^d}\bar{M}''\step{\bar{\sigma}'}\bar{M}'$ where
	$\bar{M}'\geq M'$ and $\carrier{\bar{M}'}=\carrier{M'}$.
\end{proof}

Now we prove the announced \PSPACE-hardness, by a reduction from the standard
\PSPACE-complete problem asking if a deterministic linear bounded
automaton with a two-letter tape-alphabet accepts a given word.

\begin{theorem}\label{thm1} The structural liveness problem (SLP) for \oIO nets is
    \PSPACE-hard.
\end{theorem}
\begin{proof}
We show the above announced reduction in a stepwise manner, also using informal
	descriptions that are formalized afterwards.
We thus assume a given deterministic
linear-bounded Turing machine
	\begin{equation}\label{eq:LBAw}
	\textnormal{
	$\lba=(Q,\Sigma,\Gamma,\delta,q_0,\{q\acc,q\rej\})$  and a word
		$w=x_1x_2\cdots x_n$,
	}		
\end{equation}
	where
	$\Sigma=\Gamma=\{a,b\}$, $n\geq 1$, and
	$x_i\in\Sigma=\{a,b\}$ for
$\ininter{i}{1}{n}$. W.l.o.g.\ we assume that the computation of $\lba$
	on $w$ (starting in the configuration $q_0w$)
	finishes in the accepting state $q\acc$ or the
	rejecting state $q\rej$
with the head scanning the cell $1$. We also view the transition
	function $\delta:
	(Q\smallsetminus\{q\acc,q\rej\})\times\Gamma\rightarrow
	Q\times\Gamma\times\{{-}1,{+}1\}$ as the set of
	\emph{instructions} $(q,x,q',x',m)$, writing rather
	$(q,x,q',x',m)\in\delta$ instead of $\delta(q,x)=(q',x',m)$,
	and w.l.o.g.\ we assume that $q'\neq q_0$ for each instruction
	$(q,x,q',x',m)\in\delta$; hence each computation of $\lba$,
	starting with the initial state $q_0$, never returns to $q_0$.

\begin{figure}[!h]
	\centering
        \input{figures/hardness1.tex}
	\caption{$\ins=(q,x,q',x',m)$ mimicked in $N_\lara{\lba,w}$
	(left)\  and in $N'_\lara{\lba,w}$ (right).}
        \label{fig:hardness1}
\end{figure}

\medskip
	It is straightforward
(and standard) to simulate the computation of $\lba$ on
	$w=x_1x_2\cdots x_n$
	with a 1-safe conservative Petri net $(N_\lara{\lba,w},M_0)$, as we now sketch (see
	Figure~\ref{fig:hardness1}(left)); by ``1-safe'' we mean that
	$M(p)\in\{0,1\}$ for all $M\in\rset{M_0}$ and all places
	$p$ of $N_\lara{\lba,w}$. Given~(\ref{eq:LBAw}), we construct
$N_\lara{\lba,w}$ as follows.
	\begin{itemize}	
\itemsep=0.95pt
		\item For each control state $q\in Q$ and each head-position $\ininter{i}{1}{n}$
			we create a place $p_\lara{q,i}$.
		\item For each  tape-cell $i\in[1,n]$ and each
	tape-symbol $x\in\Gamma=\{a,b\}$ we create a place
			$p_\lara{i,x}$.
		\item	 Each
instruction $\ins=(q,x,q',x',m)\in\delta$ is implemented by net
transitions $t_\lara{\ins,i}$, $i\in[1,n]$, where
			$\premset{t_\lara{\ins,i}}=\{p_\lara{q,i},p_\lara{i,x}\}$ and
	$\postmset{t_\lara{\ins,i}}=\{p_\lara{q',i+m},p_\lara{i,x'}\}$,
excluding the cases where $i{+}m\not\in[1,n]$.
(The notation stresses that the multisets $\premset{t_\lara{\ins,i}}$ and
			$\postmset{t_\lara{\ins,i}}$ are sets.)
	\end{itemize}	
A configuration $(q,i,u)$ of $\lba$, where
$q\in Q$, $\ininter{i}{1}{n}$, and
	$u=y_1y_2\cdots y_n\in\Gamma^*=\{a,b\}^*$, is mimicked by
the marking of $N_\lara{\lba,w}$
with one token in $p_\lara{q,i}$ and one
token in $p_\lara{j,y_j}$ for each $\ininter{j}{1}{n}$. The initial
	configuration $q_0w$ of $\lba$ is mimicked by the respective
	\emph{initial marking} $M_0$ in  $N_\lara{\lba,w}$, from which
	there exists only one execution, simulating the computation of
	$\lba$ on $w$.

\medskip
Since $N_\lara{\lba,w}$ is not an \oIO net in general, we transform
$N_\lara{\lba,w}$ to an \oIO net $N'_\lara{\lba,w}$ as depicted in
Figure~\ref{fig:hardness1}:
for each $\ins=(q,x,q',x',m)\in\delta$ and $i\in[1,n]$,
\begin{itemize}
\itemsep=0.9pt
	\item we add a~place $p_\lara{\ins,i}$ and
\item replace the transition
$t_\lara{\ins,i}$ with three IO transitions, namely\vspace*{-1mm}
		\begin{center}
		$\iotrans{t^{\tbegin}_\lara{\ins,i}}{p_\lara{q,i}}{p_\lara{i,x}}{p_\lara{\ins,i}}$,
$\iotrans{t^{\tmove}_\lara{\ins,i}}{p_\lara{i,x}}{p_\lara{\ins,i}}{p_\lara{i,x'}}$,
$\iotrans{t^{\tend}_\lara{\ins,i}}{p_\lara{\ins,i}}{p_\lara{i,x'}}{p_\lara{q',i+m}}$; \vspace*{-1mm}
		\end{center}
			in fact, we omit $t^{\tmove}_\lara{\ins,i}$ if $x=x'$.
\end{itemize}
It is clear that $N'_\lara{\lba,w}$ starting from $M_0$ also simulates
the computation of $\lba$ on $w$.

\medskip
Since the computation of $\lba$ always finishes with the head scanning
the cell $1$,
we get that $\lba$ accepts $w$ iff the ``state-position token''
from $p_\lara{q_0,1}$ in $M_0$ eventually
moves to the place $p_\lara{q\acc,1}$.
Hence the \emph{reachability} (and \emph{coverability}) problem for \oIO
nets is \PSPACE-hard (in fact, \PSPACE-complete).

\begin{figure}[!h]
\vspace*{-3mm}
	\centering
        \input{figures/hardness2.tex}\vspace*{-2mm}
        \caption{Marked $p\free$ allows to freely change
the mimicked configurations of $\lba$ in $N''_\lara{\lba,w}$.}
        \label{fig:hardness2}\vspace*{-2mm}
\end{figure}

\medskip
For the \emph{liveness} problem, we construct the \oIO net
$N''_\lara{\lba,w}$ arising from $N'_\lara{\lba,w}$ as depicted in
Figure~\ref{fig:hardness2}:

\begin{itemize}
\itemsep=0.9pt
	\item 	we add a place $p\run$, with a token in the initial
		marking $M_0$, and a place $p\free$, unmarked in
		$M_0$, and transitions
		$\iotrans{t_\textsc{a}}{p\run}{p_\lara{q\acc,1}}{p\free}$
		and
		$\iotrans{t'_\textsc{a}}{p\free}{p_\lara{q\acc,1}}{p\run}$;
	\item
		for each pair	$(\lara{q,i},\lara{q',i'})$ where
		$\lara{q,i}\neq\lara{q',i'}$ we add
		a transition \vspace*{-2mm}
		\begin{center}
		$\iotrans{t_{\lara{q,i,q',i'}}}{p_\lara{q,i}}{p\free}{p_\lara{q',i'}}$,\vspace*{-2mm}
        \end{center}
		and for each $(i,x,y)$ where $i\in[1,n]$,
		$x,y\in\Gamma=\{a,b\}$ and $x\neq y$ we add a transition \vspace*{-2mm}
		\begin{center}		
	$\iotrans{t_{\lara{i,x,i,y}}}{p_\lara{i,x}}{p\free}{p_\lara{i,y}}$. \vspace*{-2mm}
		\end{center}
		\end{itemize}
We can easily verify that $\lba$
accepts $w$ iff the initial marking $M_0$ of $N''_\lara{\lba,w}$
(mimicking $q_0w$ and having a token in $p\run$) is live:
A crucial
fact is that $t_{\textsc{a}}$ gets enabled iff  $\lba$
accepts $w$, after which any configuration (i.e., any marking
mimicking a configuration of $\lba$) can be repeatedly installed, which
makes each
transition live, including all
$t^{\tbegin}_\lara{\ins,i}$, $t^{\tmove}_\lara{\ins,i}$,
$t^{\tend}_\lara{\ins,i}$ (from Figure~\ref{fig:hardness1})
and also
$t_{\textsc{a}}$ due to $t'_{\textsc{a}}$
 (from Figure~\ref{fig:hardness2}).
 Hence the liveness problem is \PSPACE-hard (and
\PSPACE-complete) as well.

\begin{figure}[!b]
\vspace*{-2mm}
	\centering
        \input{figures/hardness3.tex}\vspace*{-2.5mm}
        \caption{$\bar{N}_\lara{\lba,w}$ arises from $N''_\lara{\lba,w}$  by replacing $t'_{\textsc{a}}$
        from Figure~\ref{fig:hardness2} by ``hard-wiring'' $M_0$.}\label{fig:hardness3}\vspace*{-3mm}
\end{figure}

\medskip
If $\lba$ accepts $w$, then the \oIO net $N''_\lara{\lba,w}$ is
structurally live since $M_0$ is a live marking. But to show the
\PSPACE-hardness of the \emph{structural liveness} problem (SLP) we
also need to guarantee that all markings are non-live in the
respective net if $\lba$ rejects $w$. This leads us to a final
step of our construction: the net $N''_\lara{\lba,w}$ is modified
by replacing the transition $t'_{\textsc{a}}$ (recall
Figure~\ref{fig:hardness2}) with a collection of places and
transitions that essentially ``hard-wire'' the initial marking $M_0$,
as depicted in Figure~\ref{fig:hardness3};
thus the final \oIO net $\bar{N}_\lara{\lba,w}$ arises.

\medskip
More concretely, the net $\bar{N}_\lara{\lba,w}$ arises from
$N''_\lara{\lba,w}$
by removing
$t'_{\textsc{a}}$ (with its adjacent edges) and
adding the following objects (we recall that $w=x_1x_2\cdots x_n$):
\begin{itemize}
\itemsep=0.9pt
	\item places $p_\lara{\init{i}}$ for all
    $\ininter{i}{1}{n}$, with zero tokens in $M_0$;
 \item  a transition  $\iotrans{t_\lara{\init{1}}}{p\free}{p_\lara{1,x_1}}{p_\lara{\init{1}}}$;
 \item for each $\ininter{i}{2}{n}$,	transitions\vspace*{-1.6mm}
		\begin{center}
		$\iotrans{t_\lara{\init{i}}}{p_\lara{\init{i-1}}}{p_\lara{i,x_i}}{p_\lara{\init{i}}}$
			and
			$\iotrans{t_{\lara{\trev,i}}}{p_\lara{\init{i}}}{}{p\free}$; \vspace*{-1.6mm}
		\end{center}
 \item a transition	$\iotrans{t\run}{p_\lara{\init{n}}}{p_\lara{q_0,1}}{p\run}$.
\end{itemize}
(We do not need observation-places for transitions
$t_{\lara{\trev,i}}$.)

\medskip
It is easy to check that we also have
\begin{center}
$\lba$ accepts $w$ iff $M_0$ is live in $\bar{N}_\lara{\lba,w}$.
\end{center}
Indeed: We recall that $\lba$ accepts $w$ iff $p_\lara{q\acc,1}$ gets
eventually marked when  $\bar{N}_\lara{\lba,w}$
starts with $M_0$. Then the token from $p\run$ can move to
$p\free$ by $t_{\textsc{a}}$, which allows to freely perform all
transitions that are changing the mimicked configurations of $\lba$ and
that are simulating the instructions $\ins=(q,x,q',x',m)$ from
$\delta$.
The token in  $p\free$
can go back to
$p\run$ only via performing the whole sequence
\begin{equation}\label{eq:initseq}
t_\lara{\init{1}}t_\lara{\init{2}}\cdots
t_\lara{\init{n}}t\run.
\end{equation}
Performing this sequence can be always ``aborted'' by performing one of
the transitions $t_\lara{\trev,i}$; on the other hand,
performing the whole sequence guarantees that the initial
$M_0$ (corresponding to the initial configuration $q_0w$) had been
installed before the sequence~(\ref{eq:initseq}) started; here we use the assumption that in all
$\ins=(q,x,q',x',m)\in\delta$
we have $q'\neq q_0$, which entails that during performing
the whole sequence~(\ref{eq:initseq}) no other transition could be performed
(since the single ``state-position token'' must be in $p_\lara{q_0,1}$
during the whole sequence).

\medskip
It remains to show that
\begin{center}
if $\lba$ rejects $w$, then each marking $M$ of $\bar{N}_\lara{\lba,w}$ is
	non-live.
\end{center}
For the sake of contradiction we assume
that  $\lba$ rejects $w$ and $M$ is a live marking of
$\bar{N}_\lara{\lba,w}$.
Since $M$ is live, it contains at least one token on some
$p_\lara{q,i}$, and for each $i\in[1,n]$ it contains
at least one token on one of the places $p_\lara{i,a}$ and
$p_\lara{i,b}$; moreover, from $M$
we can reach a marking in which
$p\free$ is marked and then install a~``pseudoinitial'' marking $\bar{M}_0\geq
M_0$ where $\carrier{\bar{M}_0}=\carrier{M_0}$ (the carrier of the
pseudoinitial marking coincides with the carrier of the initial
marking).
Since $\lba$ rejects $w$, for the sequence $\sigma$ of transitions of
$\bar{N}_\lara{\lba,w}$ that simulates the computation of $\lba$ on $w$ we
have $M_0\step{\sigma}M'$ where
the state-position token in $M'$ is on $p_\lara{q\rej,1}$. In this
marking $M'$ all transitions are dead: for each transition $t$ in
$\bar{N}_\lara{\lba,w}$ there is a place $p_t\in\preset{t}$ such that
$M'(p_t)=0$. (We note that the set $S=\{p_t\mid t$ is a transition
in $\bar{N}_\lara{\lba,w}\}$ is a non-empty set of places that is
unmarked in $M'$, and $S$ is a siphon.)
By Proposition \ref{prop:greatermarking} we also have
$\bar{M}_0\step{\bar{\sigma}}\bar{M}'$ where
$\carrier{\bar{M}'}=\carrier{M'}$, and thus in $\bar{M}'$ all
transitions are dead as well (the above siphon $S$ is also unmarked in
$\bar{M}'$).
Since $\bar{M}'\in\rset{M}$, we have got
a contradiction with the assumption
that $M$ is live.
\end{proof}

\section{Proof strategy for upper bounds}\label{sec:strategy}

We first recall the monotonicity property of
nets: if $M_1\step{\sigma}M_2$, then
$(M_1+M)\step{\sigma}(M_2+M)$.
This entails a simple fact: if we ignore some places (that can be
understood as marked with infinite amounts of tokens),
then any original execution can be performed in the case with
ignored places as well.

\begin{proposition}[If a transition is dead with ignored places, then it is
	dead originally]\label{prop:ignoredead}
	Let $N=(P,T,F)$ be a (general) net.
For all $M,P',t$, where $M:P\rightarrow\Nat$, $P'\subseteq P$, and
	$t\in T$, we have:
	if $t$ is dead in $(\rst{N}{(P',T)},\rst{M}{P'})$, then $t$ is
	dead in $(N,M)$ as well.
\end{proposition}	
\begin{proof}
	For any execution $M\step{\sigma}M'$ in $N$ there is obviously
	the execution
	$\rst{M}{P'}\step{\sigma}\rst{M'}{P'}$ in $\rst{N}{(P',T)}$.
	Hence if $t$ is non-dead in  $(N,M)$ (we have an execution
	$M\step{\sigma t} M'$ of $N$), then $t$ is non-dead in
 $(\rst{N}{(P',T)},\rst{M}{P'})$ as well.
\end{proof}

A crucial ingredient for proving the upper-bound results stated by
Theorem~\ref{thm:upperboundstable} in
Section~\ref{subsec:results} is captured by the following lemma, which
will be proven in Section~\ref{sec:cruclemma}.
(We use the designation ``Lemma'' rather than ``Proposition''
for the claims that we highlight as crucial for our theorems.)

\begin{lemma}[For ord-BIMO nets, if $M_0$ is non-live, then there is
	a simple witness $M\wit\in\rset{M_0}$]\label{lem:crucialingred}
A marking $M_0$ of
an ord-BIMO net $N=(P,T,F)$ is
 non-live iff
there are
	\begin{itemize}
		\item			
$M\wit\in\rset{M_0}$
	(a~witness marking),
\item
			$P\cruc\subseteq P$ (a set of crucial
	places), and
\item
			a nonempty set $T\dead\subseteq T$ (a set of dead
	transitions)
	\end{itemize}			
	such that
	\begin{enumerate}
		\item
			$|\rst{(M\wit)}{P\cruc}|<|P\cruc|$ (hence in $M\wit$ the sum of tokens
			on the places from $P\cruc$
			is at most $|P\cruc|-1$), and we can, moreover,
			require that $M\wit(p)\in\{0,1\}$ for each
			$p\in P\cruc$;
		\item
			$\rst{N}{(P\cruc,T\smallsetminus T\dead)}$ is
			an ord-IMO net (and is thus
			conservative);
		\item
all transitions from $T\dead$ are dead
			in $(\rst{N}{(P\cruc,T)},\rst{(M\wit)}{P\cruc})$.
	\end{enumerate}
\end{lemma}
We note that $\rst{N}{(P\cruc,T)}$ can have transitions $t$
for which $\premset{t}=\postmset{t}=\emptyset$; they are
live but have no effect. Technically we view such
transitions also as IMO-transitions (we might again imagine
a marked dummy place serving as both the source and the
destination).

\medskip
In a particular case demonstrating
Lemma~\ref{lem:crucialingred}, $P\cruc$ can be a siphon of
$N$ that is unmarked at $M\wit$. Recall Figure~\ref{fig:ex1}
with the siphon $\{p_2,p_3,p_4\}$ where
$T_D=\{t_1,t_2,t_3,t_4\}$.
A more general case is exemplified
by the \oBIMO net $N=(P,T,F)$ in
Figure~\ref{fig:ex4}, with $M_0$ satisfying
$M_0(p)=1$ for all $p\in P$. (The ``observation edges'' in
Figure~\ref{fig:ex4} are drawn
as dotted, for better lucidity.)
 By executing $M_0 \step{t_3t_4t_2t_1t_1t_6t_6t_6} M\wit$
we get the marking depicted in Figure~\ref{fig:ex4}, where we put
$P\cruc=\{p_1,p_2,p_3,p_4,p_5\}$ and $T\dead=\{t_1,t_6,t_7\}$.

\begin{figure}[!h]
%\vspace*{-1mm}
    \centering
        \input{figures/example1.tex}\vspace*{-3.5mm}
\caption{From $(1,1,1,1,1,1,1)$ we reach the depicted $M\wit$, with
	$P\cruc=\{p_1,p_2,p_3,p_4,p_5\}$. }
    \label{fig:ex4}\vspace*{0.5mm}
\end{figure}

We observe that the ``if'' direction of Lemma~\ref{lem:crucialingred} is
clear (by recalling Proposition~\ref{prop:ignoredead}).
We also observe that, given an ord-BIMO net $N=(P,T,F)$,
 $M\wit:P\rightarrow\Nat$,
$P\cruc\subseteq P$, and $T\dead\subseteq T$,  conditions
$1$ and $2$ of Lemma~\ref{lem:crucialingred} can be checked trivially,
while deciding condition $3$ is surely in PSPACE
(since 	$\rst{N}{(P\cruc,T\smallsetminus T\dead)}$ is conservative):
we recall Remark~\ref{rem:pspace}, and note that we can
verify that in $(\rst{N}{(P\cruc,T\smallsetminus
T\dead)},\rst{(M\wit)}{P\cruc})$ we
cannot cover $\premset{t}\cap P\cruc$ for any $t\in T\dead$.

Lemma~\ref{lem:crucialingred} (whose ``only if'' direction is proven
in Section~\ref{sec:cruclemma})
will allow us to derive the upper bounds in Table~\ref{tab:res}
for ordinary nets in Section~\ref{sec:upperord}, which
is extended
to non-ordinary nets
by a simple construction
in Section~\ref{sec:extensions}.
These upper bounds also enable to give a smooth proof that reachability
of a ``simple witness'' $M\wit$ of non-liveness of $M_0$ in BIMO nets can
be verified in polynomial space as well; this is shown in
Section~\ref{sec:inpspace}.

\section{Proof of Lemma~\ref{lem:crucialingred}}\label{sec:cruclemma}

We first sketch the organization of the proof.
We start by introducing the notion of optimal markings and an
observation that guarantees
reachability of such markings,
in the case of general live Petri nets. Then we study optimal
markings in the case of \oBIMO nets;  here we also use the notion of
``relaxed nets'', with omitted observation edges. A crucial fact is
then captured by
Lemma~\ref{lem:spread}.
Finally, another
general notion, of DL-markings where each transition is
either dead or live and at least one is dead,
leads to finishing the whole proof of Lemma~\ref{lem:crucialingred}.

\paragraph{Optimal markings, i.e.\ carrier-maximal and self-coverable
markings.}\hfill\\
Let $N=(P,T,F)$ be a net, and $M$ a marking of $N$. We define the
following notions:
\begin{itemize}
\itemsep=0.95pt
	\item $M$ is \emph{carrier-maximal}
if for each $M'\in\rset{M}$ we have $|\carrier{M'}|\leq
|\carrier{M}|$;
\item
 $M$ is \emph{self-coverable} if there
is an execution  $M\step{\sigma}M'$ where $M\leq M'$
		and $\sigma$ is \emph{full}, i.e.,
 each  $t\in T$ has at least one occurrence in $\sigma$;
\item
	$M$ is \emph{optimal} if $M$ is both carrier-maximal and
	self-coverable.
\end{itemize}
For instance, the marking $M\wit$ in Figure~\ref{fig:ex4} is optimal
for the net $N'$ arising from the depicted net $N$ by removing all
dead transitions ($t_1,t_6,t_7$) with their incident edges.

\begin{proposition}\label{prop:liveentailscmsc}
Let $M_0$ be a live marking of a net $N=(P,T,F)$.
	Then there is an optimal marking $M\in\rset{M_0}$.
\end{proposition}
\begin{proof}
Let $M_0$ be a live marking of $N=(P,T,F)$;
	we assume  $T\neq\emptyset$ (otherwise the claim is trivial).
	We consider an infinite execution
	${M}_0\step{{\sigma}_1}{M}_1\step{{\sigma}_2}{M}_2\step{{\sigma}_3}\cdots$
    where
    \begin{enumerate}
    \itemsep=0.95pt
        \item  for each $i\geq 1$, ${\sigma}_i$ is
		    full (contains all $t\in T$), and
        \item for each $i\geq 1$ we have
              $\sizeof{\carrier{{M}_{i}}}=
              \max\,\{\,\sizeof{\carrier{M}};
		    \gt{{M}_{i-1}}{\sigma}{M}$ for some full
              $\sigma\}$.
    \end{enumerate}
	Such an execution obviously exists since $M_0$ is live (which
	entails that all $M\in\rset{M_0}$ are live).

\medskip
	We observe that $M_i$ are carrier-maximal for all $i\geq 1$: if we had
	$M_i\step{\sigma}M$ where $|\carrier{M_i}|<|\carrier{M}|$,
	then $M_{i-1}\step{\sigma_i\sigma}M$ would violate the
	condition $2$ (since $\sigma_i\sigma$ is full).
	By Dickson's lemma
	there are $i_1, i_2$ such that $1\leq i_1<i_2$
	and ${M}_{i_1}\leq{M}_{i_2}$; hence
 $M_{i_1}$ is both
	carrier-maximal and self-coverable
	(which also entails that
	$|\carrier{M_{i_1}}|=|\carrier{M_{i_2}}|$).
\end{proof}

	\paragraph{The relaxed net $\prlx{N}$ associated with an
	\oBIMO net $N$.} Let $N = (P,T,F)$ be an \oBIMO net, and let
	$E\subseteq (P\times T)\cup (T\times P)$ be the set of
	its ``moving'' edges, i.e.\ the least set such that each
	transition
	$\iotrans{t}{p_s}{\{p_{o_1},\dots,p_{o_\ell}\}}{\{p_{d_1},\dots,p_{d_k}\}}$
	entails that the edges
	$(p_s,t)$ and $(t,p_{d_j})$ ($j\in[1,k]$) are in $E$.
	The \emph{relaxed net related to} $N$ is the subnet
\begin{center}	
	$\prlx{N}	= (P,T,F')$
\end{center}
	of $N$
where for all $(x,y)\in
	(P\times T)\cup (T\times P)$ we have $F'(x,y)=1$ if
	$(x,y)\in E$, and $F'(x,y)=0$ otherwise.
(We recall that we assume that each transition $t$ has a fixed
source place $p_s$, and
that here we deal with \emph{sets}
$\{p_{o_1},\dots,p_{o_\ell}\}$
and $\{p_{d_1},\dots,p_{d_k}\}$ since the considered nets are
ordinary.)

\medskip

\begin{figure}
    \centering
        \input{figures/example2.tex}
    \caption{A structurally live \oBIMO net.}
    \label{fig:ex3}
\end{figure}
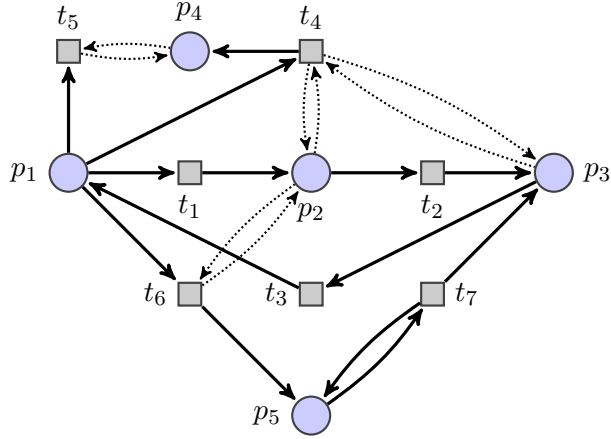

For instance, in the nets depicted in
Figures~\ref{fig:ex4},~\ref{fig:ex3}, and~\ref{fig:ex2}
the thick edges are those remaining in the respective
relaxed nets.
We note that the edges between $p_5$ and $t_7$
in Figure~\ref{fig:ex3} are thick, since $t_7$ is understood as
$\obiotrans{t_7}{p_5}{}{p_5,p_3}$.

\medskip
Informally speaking, $\prlx{N}$ represents the behaviour of $N$ when the
observation-condition is relaxed, which means that a transition can fire even
if its observation places are unmarked.

\begin{remark}
For \oIMO nets
their relaxation nets
are $S$-nets in terminology of, e.g.,
\cite{desel_esparza_1995} (pre-msets and post-msets of all transitions
are singletons).
Generally, for \oBIMO nets
their relaxation nets
are BPP-nets in terminology of, e.g.,
\cite{DBLP:journals/fuin/Esparza97} (pre-msets of all transitions
are singletons).
\end{remark}

\paragraph{Nets $\prlx{N}$ as directed bipartite graphs, paths,
(bottom, top) components.}
We recall that any ordinary net $N=(P,T,F)$ (with $F:(P\times
    T)\cup(T\times P)\rightarrow \{0,1\}$) can be naturally
    viewed as a directed bipartite graph where $P\cup T$ is
    the set of vertices and  $\{(x,y)\in(P\times
    T)\cup(T\times P)\mid F(x,y)=1\}$ is the set of edges.
    We thus use the standard notions like \emph{a path in} $\prlx{N}$ or
    \emph{a~strongly connected component (scc) of}
    $\prlx{N}$.
	By a~\emph{proper successor of} an \emph{scc} $C$ we mean an
	scc $C'$ such that $C'\neq C$ and
there is a path from $C$ to $C'$ (in $\prlx{N}$).
An scc $C$ is a \emph{bottom component} if there is
no proper successor of $C$;
$C$ is a \emph{top component} if there is no $C'$ such that $C$ is
a proper successor of $C'$.

    For an scc $C$, by $P_C$ we denote the set of places
    in $C$ (which is empty if $C$ consists just of one transition).
    We also say just ``a component'' instead of ``an scc''.

\medskip
For instance, $\prlx{N}$ related to $N$ in
Figure~\ref{fig:ex4} has just one (strongly connected) component,
in Figure~\ref{fig:ex3} we have four components,
three of them being trivial (i.e.\ singletons), namely $\{t_5\}, \{p_4\},
\{t_4\}$,
and in Figure~\ref{fig:ex2} we have two components,
in this case both
nontrivial.

\paragraph{Rich and poor components in marked ord-BIMO nets.}
We assume a given \oBIMO net $N=(P,T,F)$.
For a marking $M:P\rightarrow\Nat$, we say that
\[
	\textnormal{an \emph{scc} $C$ of $\prlx{N}$
is }
\begin{cases}
	\textnormal{\emph{ rich at} } M \ \dots \  \textnormal{if } M(P_C) \geq
	\sizeof{P_C},\\
	\textnormal{\emph{ poor at} } M \ \dots \  \textnormal{if } M(P_C) <
	\sizeof{P_C}.\\
\end{cases}	
\]
If $C$ is a rich (or poor) scc in $\prlx{N}$ at $M$, then
we also say that $C$ is a \emph{rich} (or \emph{poor}) \emph{component
in} $(N,M)$. We note that any component consisting of a single
transition is always rich.

\medskip
For instance, the only component of $\prlx{N}$ depicted in
Figure~\ref{fig:ex4} is rich at both $(1,1,1,1,1,1,1)$ and $M\wit$.
We recall the net $N'$ arising by removing the transitions $t_1,t_6,t_7$
(from the net $N$ depicted in Figure~\ref{fig:ex4}),
and note that all
transitions of $N'$ are live at $M\wit$. We have three components in
$(N',M\wit)$ (one of them consisting just of $p_5$), one rich and two
poor. We also recall that $M\wit$ is an optimal marking in $N'$.

\medskip
Now we show a crucial fact (called ``Lemma''
rather than ``Proposition'' in line with our remark
before the statement of Lemma~\ref{lem:crucialingred}).

\begin{lemma}[In optimal markings, tokens are ``spread'' and
	poor components are on top]\label{lem:spread}
	Let $N=(P,T,F)$ be an ord-BIMO net,
and $M_0$ an optimal marking.
	Then:
	\begin{enumerate}
\itemsep=0.9pt
		\item for each rich component $C$ in $(N,M_0)$ we have $M_0(p)\geq 1$ for all
	$p\in P_C$;
\item  for each poor component $C$  in $(N,M_0)$ we have
	$M_0(p)\in\{0,1\}$ for all $p\in P_C$;
\item for each poor component $C$ in $(N,M_0)$ it holds that $C$
			is a top component	of $\prlx{N}$, and that $\rst{N}{(P_C,T)}$ is an
			\oIMO net (i.e., each transition in $C$ has the source
			place in $P_C$ and precisely one destination place in
			$P_C$).
	\end{enumerate}
	\end{lemma}	
\begin{proof}
	Let $N = (P,T,F)$, $M_0$ satisfy	the assumptions.

\medskip
For the sake of contradiction we assume that the statement does not
	hold, and we choose 	a~component $C_0$ (an scc of $\prlx{N}$) such that $C_0$
	violates some of  conditions $1-3$, while 	all proper successors of $C_0$ (if some exist) do not violate
	these conditions. Each proper successor $C$ of $C_0$ is thus a
	rich component and satisfies 	$M_0(p)\geq 1$ for all $p\in P_{C}$ (which is trivial when
	$P_{C}=\emptyset$, i.e., 	when $C$ 	consists of a single 	transition).
By $P_\textsc{succ}$ we denote the union of the sets $P_{C}$ for all proper successors $C$ of $C_0$
	(hence $M_0(p)\geq 1$ for all $p \in P_\textsc{succ}$).

\medskip
We fix
	two places $p_1,p'_1 \in P_{C_0}$ such that
	\begin{equation}\label{eq:poneponeprime}
		{M}_0(p_1)\geq 2\textnormal{ and }
		{M}_0(p'_1)=0;
	\end{equation}	
this is obviously possible in the case when $C_0$ is rich and violates
	condition $1$ as well as in the case when
	$C_0$ is poor and violates condition $2$.

\medskip
Since $M_0$ is self-coverable, we can fix a full sequence $\sigma\in T^*$
	(containing all transitions from $T$) such that
	$M_0\step{\sigma}\bar{M}$ where $M_0\leq \bar{M}$;
	we note that $\carrier{M_0}=\carrier{\bar{M}}$, since $M_0$ is
	carrier-maximal.
	We can easily verify that we also have
	\begin{equation}\label{eq:sigmaone}
		\textnormal{${M}_0\step{\sigma_1}{M}_{1}$ and
		$M_0\leq M_{1}$ (where $\carrier{M_0}=\carrier{M_1}$)}
	\end{equation}
for $\sigma_1$ arising 	from $\sigma$ by omitting all occurrences of
	transitions contained in the proper successors of $C_0$.
This follows from the facts that $M_0(p)\geq 1$ for all $p\in
	P_\textsc{succ}$, 	and that all destination places of 	transitions in the proper successors of $C_0$ are in
	$P_\textsc{succ}$; hence $\sigma_1$ is performable from $M_0$,
	keeping each place in 	$P_\textsc{succ}$ marked during the whole execution
	${M}_0\step{\sigma_1}{M}_{1}$.

\medskip
We get the desired contradiction in both
 cases C1 and C2 below; the cases cover all possibilities, since each
 transition in $C_0$ has at least one destination place in $P_{C_0}$
(which follows from the fact that $C_0$ does not consist of a single
transition). \vspace*{-1mm}

		\paragraph{Case C1.}		\emph{The (violating) scc $C_0$ is not a top component of
			$\prlx{N}$, or there is a transition in $C_0$
			that has at
			least two destination places in $P_{C_0}$.}

\medskip

		In this case there is some $p_0\in P_{C_0}$ such that
			$M_0(p_0)<M_{1}(p_0)$, referring
			to~(\ref{eq:sigmaone}), since no transition in
			$\sigma_1$ decreases the token count in $C_0$
			while at least one transition in $\sigma_1$
			increases this count.
			Since $C_0$ is a strongly connected component of
			$\prlx{N}$, there
must be a path from $p_0$
			to $p'_1$ in $C_0$ (referring to $p'_1$
			in~(\ref{eq:poneponeprime}));
let the sequence of the transitions on this path be
$t'_1t'_2\cdots t'_m$.
Now we can consider the execution
\begin{center}
	$M_0\step{\sigma_1}M_1\step{\sigma_{1,1}}\step{\sigma_{1,2}}\cdots
	\step{\sigma_{1,m}}M'$
\end{center}
where $\sigma_{1,j}$ arises from $\sigma_1$ by replacing an
occurrence of $t'_j$ with $t'_jt'_j$ (the second occurrence
of $t'_j$ moves the ``excessive'' token on its way from
$p_0$ towards $p'_1$, maybe also generating additional tokens on the way).

\medskip
We clearly have
$\carrier{M'}\supseteq \carrier{M_0}\cup\{p'_1\}$,
which contradicts the assumption that $M_0$ is carrier-maximal.\vspace*{-1mm}

\paragraph{Case C2.}	\emph{The (violating) scc $C_0$ is a top component and
			each transition $t$ in
			$C_0$ has (the source place and) precisely
			one destination place in $P_{C_0}$.}

\medskip
Referring to the execution ${M}_0\step{\sigma_1}{M}_{1}$ from~(\ref{eq:sigmaone}),
we observe that we have
		$M_0(p)=M_{1}(p)$ for all $p\in P_{C_0}$ (in our
		case C2). Though
			we have $M_0(p_1)\geq 2$, and there is a path
			from $p_1$ to $p'_1$ (we refer to the places
			from~(\ref{eq:poneponeprime})), we cannot
			immediately say that a token on $p_1$ is
			``excessive''; to use an analogous idea as in C1,
we first adjust $\sigma_1$ by a careful omission of some transition occurrences, to
make clear that we actually get an excessive token after all.

\medskip
In the respective ``omitting'' construction,
for $t$ in $C_0$ we use the expression  $p\rightarrow t$ to denote
that  $p$ is the source place of $t$ (we necessarily have $p\in
P_{C_0}$);
the expression
$t\rightarrow p$ denotes that $p$ is the destination
	place of $t$ that belongs to $P_{C_0}$.

\medskip
Now we imagine constructing a sequence
\begin{equation}\label{eq:seqsigmaipi}
seq = (\sigma_1,p_1),(\sigma_2,p_2),\dots,(\sigma_k,p_k)
\end{equation}
	where 	$M_0\step{\sigma_1}M_1$
	(recall~(\ref{eq:sigmaone})),
and $p_1$ is the place from~(\ref{eq:poneponeprime}),
hence 	$M_0(p_1)\geq 2$.
Given $(\sigma_i,p_i)$, we construct $(\sigma_{i+1},p_{i+1})$ as
follows:

\begin{itemize}
\itemsep=0.95pt
     \item if there is no transition $t$ in
				$\sigma_i$ where
              $t\rightarrow p_i$, then halt (and
              put $k=i$);
        \item
(otherwise) we write $\sigma_i=\sigma't\sigma''$
where $t\rightarrow p_i$ and there is no $t'$ in $\sigma''$ for which
$t'\rightarrow p_i$, and put  $\sigma_{i+1}=\sigma'\sigma''$
				(the last transition occurrence
				putting a token in $p_i$ has been
				omitted);
              the
				source place of $t$  is taken as $p_{i+1}$
				(hence we have
				$p_{i+1}\rightarrow t\rightarrow p_i$).
\end{itemize}

We show that the construction of~(\ref{eq:seqsigmaipi}) keeps the following
conditions $1-3$
(for all $i=1,2,\dots,k$);
we recall that
$P_\textsc{succ}$ denotes the set of places of
the proper successors of $C_0$.
   \begin{enumerate}
   \itemsep=0.95pt
	 \item $M_0\step{\sigma_i}M_i$ (for some $M_i$) where
 $\carrier{M_0}=\carrier{M_i}$ and $M_i(p_i)\geq 2$;
\item 	$p_i\in P_{C_0}$ (and we might have $p_i=p_{i'}$ for
              $i\neq i'$);
\item
 	\begin{enumerate}
\item each place from $P_{\textsc{succ}}$ is marked during the whole execution
$M_0\step{\sigma_i}M_i$;
\item for each $p\in P\smallsetminus
 (P_{\textsc{succ}}\cup\{p_1,p_i\})$ we have
$M_i(p)=M_1(p)$   (hence 	$M_i(p)=M_0(p)$ for all
		$p\in  P_{C_0}\smallsetminus\{p_1,p_i\}$);
                  \item if $p_i=p_1$, then
			  $M_i(p_1)=M_0(p_1)$;
                  \item if $p_i\neq p_1$, then
		$M_i(p_1)=M_0(p_1)-1\geq 1$ and
			      $M_i(p_i)=M_0(p_i)+1$.
              \end{enumerate}
 \end{enumerate}
 The conditions surely hold for $i=1$, since
		   $M_0\step{\sigma_1}M_1$ and $M_1(p)=M_0(p)$ for all
		   $p\in P_{C_0}$; the validity of condition $3$(a)
		   was noted around~(\ref{eq:sigmaone}).
		   We also note that $p_k=p_1$
		   (in~(\ref{eq:seqsigmaipi})), since
in the case $p_i\neq p_1$ we have
		    $M_0\step{\sigma_i}M_i$ and
		   $M_i(p_i)=M_0(p_i)+1$, which entails that there is $t$ in
		   $\sigma_i$ such that $t\rightarrow p_i$.

\medskip
Now we assume that the conditions hold for
		   $i\in[1,k{-}1]$,
and show that they hold for $i+1$ as well.
  We recall that $\sigma_{i+1}=\sigma'\sigma''$ where
  $\sigma_i=\sigma't\sigma''$, $t\rightarrow p_i$, and there is
no $t'$ in $\sigma''$ for which  $t'\rightarrow p_i$.
		   Let us write
		   \begin{center}
		   $M_0\step{\sigma_i}M_i$ as
		   $M_0\step{\sigma'}\bar{M}_1\step{t}\bar{M}_2\step{\sigma''}M_i$.
		   \end{center}			
Since the number of tokens on $p_i$ never increases
in the execution   $\bar{M}_2\step{\sigma''}M_i$
and $M_i(p_i)\geq 2$,
the place $p_i$ is marked with at least two
		   tokens during the whole execution
		   $\bar{M}_2\step{\sigma''}M_i$.
We verify that $\sigma''$ is enabled
at $\bar{M}_1$:
we recall that the destination places of $t$ constitute a~subset
of $\{p_i\}\cup P_{\textsc{succ}}$,
$\bar{M}_1(p)\geq 1$ for all $p\in
		   P_{\textsc{succ}}$ (by condition $3$(a)),   and
there is no transition in $\sigma''$ with the source place in
$P_{\textsc{succ}}$ (since $\sigma''$ is a subsequence of $\sigma_1$
defined in~(\ref{eq:sigmaone})).

\medskip
In the execution
		   $M_0\step{\sigma'}\bar{M}_1\step{\sigma''}M_{i+1}$
		   (which is   ${M}_0\step{\sigma_{i+1}}M_{i+1}$)
the place $p_i$ is marked with at least one token during the whole
		   segment $\bar{M}_1\step{\sigma''}M_{i+1}$,
		   while each place from $P_{\textsc{succ}}$ is marked
		   during the whole execution
		   ${M}_0\step{\sigma_{i+1}}M_{i+1}$.

\medskip
Since 	$M_{i+1}=({M}_i-\postmset{t})+\premset{t}$, we have:
\begin{itemize}
\itemsep=0.95pt
	\item	 if $p_i=p_{i+1}$ (i.e., $p_i\rightarrow t\rightarrow p_i$), then
$\rst{(M_{i+1})}{(P\smallsetminus P_{\textsc{succ}})}=
\rst{(M_{i})}{(P\smallsetminus P_{\textsc{succ}})}$;
\item 	if $p_i\neq p_{i+1}$, then
		\begin{itemize}
			\item		$M_{i+1}(p)=M_i(p)$ for all
		$p\in P\smallsetminus
				(P_{\textsc{succ}}\cup\{p_i,p_{i+1}\})$,
	\item	 $M_{i+1}(p_i)=M_i(p_i)-1$, and
    \item $M_{i+1}(p_{i+1})=M_i(p_{i+1})+1$.
		\end{itemize}
\end{itemize}
		Since $M_i(p_i)\geq 2$, and thus $M_{i+1}(p_i)\geq 1$,
we have $\carrier{M_{i}}\subseteq\carrier{M_{i+1}}$. The condition $\carrier{M_i}=\carrier{M_0}$
and the assumption that $M_0$ is carrier-maximal thus entail that
$\carrier{M_{i+1}}=\carrier{M_0}$; hence $M_{i}(p_{i+1})\geq 1$ and  $M_{i+1}(p_{i+1})\geq 2$.

\medskip
The validity of conditions $1-3$ is thus clear for all $i\in[1,k]$.  We note that the
condition $\carrier{M_i}=\carrier{M_0}$ entails that $M_i(p'_1)=M_0(p'_1)=0$  (referring to $p'_1$
from~(\ref{eq:poneponeprime})), and thus $p_i \neq p_1'$.

The final element $(\sigma_k,p_k)$ of the sequence~(\ref{eq:seqsigmaipi})
satisfies $p_1=p_k$, $\gt{M_0}{\sigma_k}M_k$ where $M_0\leq M_k$, and
$p_1$ cannot be a destination place, and neither the
source place, of any transition in $\sigma_k$. Since
$M_0(p_1)\geq 2$, we get that $\gt{M'_0}{\sigma_k}{M'_k}$ where
$M'_0$ arises from $M_0$, and $M'_k$ from $M_k$, by removing a token
on $p_1$.
This entails that there are executions $\gt{M_i}{\sigma_k}{\bar{M}_i}$
(where $M_i\leq\bar{M}_i$)
for all
$i=1,2,\dots,k$ (since $M_i$ arises from $M'_0$ by adding a token on $p_i$).

\medskip
Since $\sigma_k$ has arisen from $\sigma_1$, in which all
transitions from $C_0$ occur,
for some $i\in[1,k]$ there must be a path from $p_i$
to $p'_1$ (in $C_0$) such that all transitions on the path occur
in $\sigma_k$; let the sequence of the transitions on this path be
$t'_1t'_2\cdots t'_m$. Now we can consider the execution
\begin{center}
	$M_0\step{\sigma_i}M_i\step{\sigma_{k,1}}\step{\sigma_{k,2}}\cdots
	\step{\sigma_{k,m}}M'$
\end{center}
where $\sigma_{k,j}$ arises from $\sigma_k$ by replacing an
occurrence of $t'_j$ with $t'_jt'_j$ (the second occurrence
of $t'_j$ moves the ``excessive'' token on its way from
$p_i$ towards $p'_1$). We clearly have
$\carrier{M'}\supseteq\carrier{M_i}\cup\{p'_1\}=\carrier{M_0}\cup\{p'_1\}$,
which contradicts the assumption that $M_0$ is carrier-maximal.
The proof is thus finished.
\end{proof}

Now we recall a natural notion and a simple fact (for general nets).\vspace*{-1mm}

\paragraph{DL-markings.}
A marking $M$ of a (general) net $N$ is a \emph{DL-marking} if
    each transition is either dead or live at $M$ and at
    least one transition is dead at $M$.

\begin{proposition}[From a non-live marking we can reach a DL-marking]\label{prop:dlmarking}\hfill\newline

 \vspace*{-12mm}
	\begin{enumerate}
\itemsep=0.9pt
		\item  A marking $M$ of a net $N$ is non-live  iff there is a
              DL-marking $M\dl\in\rset{M}$.
      \item  Given a net $N=(P,T,F)$, if $M\dl$ is
a DL-marking of $N$ and $T_L\subseteq T$ is the set of transitions that
			are live at $M\dl$, then  $M\dl$ is live in the net  $\rst{N}{(P,T_L)}$.
	\end{enumerate}
	\end{proposition}

\begin{proof}
The claim $2$ and the ``if'' part of the claim $1$ are trivial, hence
it remains to look
at the ``only-if'' part of the claim $1$.
	If $t$ is a transition that is non-live at $M$, then there is
    $M'\in\rset{M}$ where $t$ is dead; if $t'$ is a
     transition that is non-live at $M'$, then there is
    $M''\in\rset{M'}$ where $t'$ is dead, as well as $t$. By
    repeating this reasoning we arrive at some DL-marking
    $M\dl\in\rset{M}$, for any non-live $M$.
\end{proof}

\bigskip\noindent \textbf{Proof of Lemma~\ref{lem:crucialingred} (the ``only-if'' direction).}

\smallskip\noindent
Let $M_0$ be a non-live marking of
an ord-BIMO net $N=(P,T,F)$.
By Proposition~\ref{prop:dlmarking}(1) there is a DL-marking
in $\rset{M_0}$;
by
Propositions~\ref{prop:dlmarking}(2) and~\ref{prop:liveentailscmsc}
there is even a DL-marking
$M\dl\in\rset{M_0}$
that is optimal (carrier-maximal and
self-coverable) in $\rst{N}{(P,T_L)}$ where $T_L$ is the set
of transitions of $N$ that are live at $M\dl$. (The marking
$M\wit$ depicted in Figure~\ref{fig:ex4} is an example of such
marking.)

\medskip
Let us choose such an optimal $M\dl$,
and let $P\cruc$ be the union of places from the strongly connected components of
$\prlx{\rst{N}{(P,T_L)}}$ that are poor at $M\dl$ (hence we surely
have $|\rst{(M\dl)}{P\cruc}|<|P\cruc|$).
By
Lemma~\ref{lem:spread} we derive that $M\dl(p)\geq 1$ for all
$p\in P\smallsetminus P\cruc$, $M\dl(p)\in\{0,1\}$ for all $p\in
P\cruc$, and
	$\rst{N}{(P\cruc,T_L)}$ is an \oIMO net.
	We also note that each transition
$\iotrans{t}{p_s}{\{p_{o_1},\dots,p_{o_\ell}\}}{\{p_{d_1},\dots,p_{d_k}\}}$
where $p_s\in P\smallsetminus P\cruc$ satisfies
$\{p_{d_1},\dots,p_{d_k}\}\subseteq P\smallsetminus P\cruc$
(if the source place is outside $P\cruc$, then all destination places are
outside $P\cruc$).

\medskip
To finish proving Lemma~\ref{lem:crucialingred}, it suffices
to show
that the transitions from the set
$T\dead=T\smallsetminus T_L$ (consisting of the transitions that are dead
at $M\dl$ in $N$) are dead in
$(\rst{N}{(P\cruc,T)},\rst{(M\dl)}{P\cruc})$.
For the sake of contradiction we assume
that $\rst{(M\dl)}{P\cruc}\step{\sigma}M$ ($M:P\cruc\rightarrow \Nat$)
in $\rst{N}{(P\cruc,T)}$ where $\sigma\in (T_L)^*$ and
$M\geq(\premset{t}\cap P\cruc)$ for some $t\in T\dead$.
We recall that $M\dl(p)\geq 1$ for all $p\in P\smallsetminus P\cruc$,
hence for
$\sigma'$ arising from $\sigma$ by omitting
the transitions whose source place is outside $P\cruc$
(and whose destination places are thus outside $P\cruc$ as well)
we get
$M\dl\step{\sigma'}M'$ in $N$. We have $M'\geq \premset{t}$,
since $M'$ coincides with $M$ on $P\cruc$, and $M'(p)\geq 1$ for all
$p\in P\smallsetminus P\cruc$; but this contradicts
the assumption that $t$ is dead at $M\dl$ in $N$.
\hfill $\square$

\section{Upper bounds for ordinary nets}\label{sec:upperord}

In Section~\ref{sec:liveinoBIMO} we show an important lemma that
clarifies the main bounds in Table~\ref{tab:res} (in
Section~\ref{subsec:results}) for ordinary nets; to this aim
 we also have a more detailed look at executions by viewing
 tokens as individual objects.
In Section~\ref{sec:oIMOcase} we give the strengthened bound for
structurally live ord-IMO
nets, showing that here there are live markings having at most one token in
each place; in Section~\ref{sec:counterBIO} we show that this does not hold for
ord-BIO nets.

\subsection{Live markings in ord-BIMO nets}\label{sec:liveinoBIMO}

In this section we prove the next lemma, which has a
straightforward consequence for live markings in ord-BIMO nets.

\begin{lemma}\label{lem:liveP}
Let $N=(P,T,F)$ be an ord-BIMO net, and let $M,M'$ be two markings of $N$ that
coincide except of one place $p_0$ for which we have
	$2|P|-1\leq M(p_0)=M'(p_0)-1$. Then $M$ is live iff $M'$ is 	live.
\end{lemma}

\begin{corollary}
Let $N=(P,T,F)$ be an ord-BIMO net. The set $\calL_N$ of live markings of $N$ is
	determined by its finite (basic) subset \vspace*{-1mm}
	\begin{center}
	$\calB_N=\{M\in\calL_N\mid M(p)\leq B$ for all $p\in
	P\}$ where $B=2|P|{-}1$,\vspace*{-1mm}
	\end{center}
		since \vspace*{-1mm}
	\begin{center}
		$\calL_N=\{M\in\Nat^P\mid M'\in\calB_N$ where
	$M'(p)=M(p)$ if $M(p)\leq B$, and $M'(p)=B$
	otherwise$\}$. \vspace*{-1mm}
	\end{center}		
\end{corollary}

\begin{remark}
We note that liveness is not generally monotonic, even for \oIO
nets. Figure~\ref{fig:ex2} shows an example of an \oIO net with a  live
marking (as can be easily verified);
adding a~token to $p_2$ yields a non-live marking, since in this case firing $t_5t_2t_6$ makes all
transitions dead.
Another example is
the \oBIMO net in Figure~\ref{fig:ex3}; we can check that
the markings $(1,0,0,0,0)$ and $(3,0,0,0,0)$ are non-live, while
$(2,0,0,0,0)$ and  $(3,0,0,0,1)$ are live.
\end{remark}

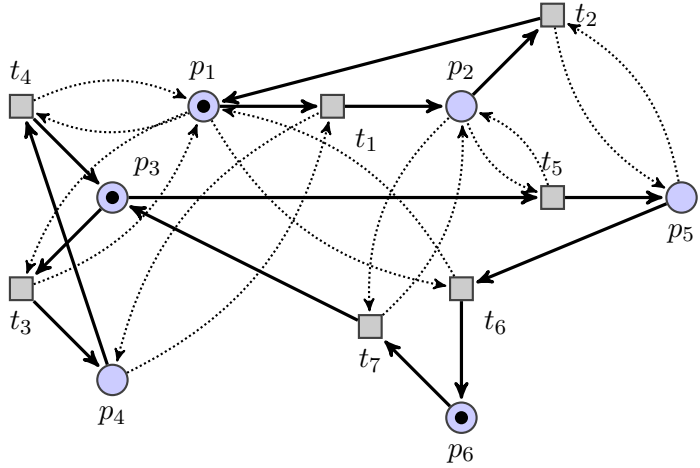
\begin{figure}
	\centering
        \input{figures/example3.tex}
    \caption{A live marking of an \oIO net $N$; adding a token on
	$p_2$ makes it non-live.}
    \label{fig:ex2}
\end{figure}

\paragraph{Executions with individual tokens, ID-valuations,
relations $\imsucct{t}$, $\imsucc$, $\imsucc^*$.}\hfill \smallskip\\
Let us have  a more detailed look at a fixed execution
	\begin{equation}\label{eq:execIDs}
M_0\step{t_1}M_1\step{t_2}M_2\cdots\step{t_m}M_m
	\end{equation}
for an ord-BIMO net $N=(P,T,F)$. Now we view tokens as individual objects; each
token is identified with its unique ID from an infinite domain $\idset$ (e.g.,
we can have $\idset=\Nat$). A~marking $M:P\rightarrow \Nat$
thus has a related \emph{id-marking} $\idm:P\rightarrow \PF(\idset)$, where the finite sets
$\idm(p)$ and $\idm(p')$ of IDs are disjoint when $p\neq p'$
(and $\sizeof{\idm(p)}=\sizeof{M(p)}$ for all $p\in P$).
We also define the notion of \emph{an ID-valuation of the execution}~(\ref{eq:execIDs}), as
a sequence \vspace*{-1mm}
\begin{center}
$\idm_0\step{t_1}\idm_1\step{t_2}\idm_2\cdots\step{t_m}\idm_m$ \vspace*{-1mm}
\end{center}
that is stepwise created as follows:
We start with attaching a unique ID to
each token in the initial marking $M_0$, getting $\idm_0:P\rightarrow
\PF(\idset)$.
For $i\in[1,m]$, $\idm_{i}$ arises from $\idm_{i-1}$ by the following
change, based on the transition
	$\iotrans{t_{i}}{p_s}{\{p_{o_1},\dots,p_{o_\ell}\}}{\{p_{d_1},\dots,p_{d_k}\}}$:
	\begin{itemize}	
		\item
			one (arbitrarily chosen) ID
			$\textsc{i}_0\in\idm_{i-1}(p_s)$
			is removed
			(from the source place $p_s$),
and it is not used anymore (i.e.,
			$\textsc{i}_0$ will not occur in $\idm_j(p)$ for
			any $j\in[i,m]$ and $p\in P$);
\item			
		for each $j\in[1,k]$, a fresh ID $\textsc{i}_j$ (so far not occurring in the created
			ID-valuation)
			is added to
			the destination place $p_{d_j}$;
\item			
moreover, we say that each of the ``destination IDs'' $\textsc{i}_j$
is an 	\emph{immediate successor of} the ``source ID''
			$\textsc{i}_0$, which is denoted by
			$\textsc{i}_0\imsucc\textsc{i}_{d_j}$;
			in more detail we also write
			$\textsc{i}_0\imsucct{t_i}\textsc{i}_{d_j}$.
	\end{itemize}
(The ``observation IDs'' on the places $p_{o_j}$ used by $t_i$ are not
changed, being the same
 in $\idm_{i}$ as in $\idm_{i-1}$.) We note that if
	the source token from $p_s$ returns, i.e.\
	$p_s\in\{p_{d_1},\dots,p_{d_k}\}$, then the returning token
	has a new ID, which is an immediate successor of the original
	one.

We have thus defined what we mean
by an ID-valuation
$\idm_0\step{t_1}\idm_1\step{t_2}\idm_2\cdots\step{t_m}\idm_m$
of the execution~(\ref{eq:execIDs}),
while we have also introduced the \emph{immediate-successor relation}
$\imsucc$ on the set of IDs used in this ID-valuation.
We also define the \emph{successor relation} $\imsucc^*$, as the reflexive and transitive
	closure of $\imsucc$.

The relation $\imsucc^*$ is clearly a partial order that can be
presented by a set of (disjoint) trees where the roots constitute the set
$\bigcup_{p\in P}\idm_0(p)$; this is captured by the next
observation.
\begin{fact}[Each used ID is a successor of a unique ID from $\bar{M}_0$]\label{obs:IDtree}

\noindent 	Given an ID-valuation
$\idm_0\step{t_1}\idm_1\step{t_2}\idm_2\cdots\step{t_m}\idm_m$
	(of an execution
	$M_0\step{t_1}M_1\step{t_2}M_2\cdots\step{t_m}M_m$),
	for each triple $i\in[0,m]$, $p\in P$, and
	$\textsc{i}\in \idm_i(p)$
	there is a~unique place $p'\in P$ and
a~unique ID	$\textsc{i}'\in\idm_0(p')$
	such that $\textsc{i}'\imsucc^*\textsc{i}$.
\end{fact}
The next two propositions entail Lemma~\ref{lem:liveP}.\vspace*{-1mm}

\begin{proposition}\label{prop:nonlivepmar}
    Let $N=(P,T,F)$ be an \oBIMO net, and let $M_0,M'_0$ be two markings of $N$ that
coincide except of one place $p_0$ for which we have
	$|P|\leq M_0(p_0)=M'_0(p_0)-1$. If $(N,M_0)$ is non-live then
	$(N,M'_0)$ is
	non-live as well.
\end{proposition}

\begin{proof}
Let $N,M_0,M'_0,p_0$ be as in the statement, and let us assume
	that $(N,M_0)$ is non-live; we aim to show that $(N,M'_0)$ is
	non-live as well.

\smallskip
By Lemma~\ref{lem:crucialingred} there are an execution $M_0\step{\sigma}M\wit$, $P\cruc\subseteq P$, and  a
	nonempty set $T\dead\subseteq T$ 	such that 		$|\rst{(M\wit)}{P\cruc}|<|P\cruc|$,
and all  transitions from $T\dead$ are dead 		in $(\rst{N}{(P\cruc,T)},\rst{(M\wit)}{P\cruc})$ (and thus also in
		$(N,M\wit)$ by Proposition~\ref{prop:ignoredead}).
We call the \emph{places in $P\cruc$ crucial}, while 	the places 	in $P\smallsetminus P\cruc$ are \emph{don't care
		places}.	

\medskip
Let the execution
		\begin{center}
		$M_0\step{\sigma}M\wit$ be of the
form $M_0\step{t_1}M_1\step{t_2}M_2\cdots\step{t_m}M_m=M\wit$,
		\end{center}
			and
let us fix an ID-valuation
		\begin{equation}\label{eq:IDvalorigexec}
\idm_0\step{t_1}\idm_1\step{t_2}\idm_2\cdots\step{t_m}\idm_m
		\end{equation}
	of this execution. We aim to show an execution $M'_0\step{\sigma'}M'$ where
		$\rst{(M')}{P\cruc}=\rst{(M\wit)}{P\cruc}$, which 	demonstrates that $(N,M'_0)$ is non-live.

\medskip
		We say that an ID $\textsc{i}$ in $\idm_0(p)$ is	\emph{black} (or \emph{crucial}) if
			there is a place $p'\in P\cruc$ and some  $\textsc{i}'\in\idm_m(p')$ such that
			$\textsc{i}\imsucc^*\textsc{i}'$ 	(hence if the initial individual token
			$\textsc{i}$ has a successor on a~crucial place 	in $M\wit$).
Since there are at most $|P|-1$ tokens on the crucial places in
	$M\wit$, by Fact~\ref{obs:IDtree} we deduce that there are at most
	$|P|-1$ black IDs in $\bigcup_{p\in P}\idm_0(p)$.
Since $M_0(p_0)\geq |P|$, at 	least one ID in $\idm_0(p_0)$ is not black;
we fix one non-black ID $\textsc{i}_0\in\idm_0(p_0)$, and we further view
		$\textsc{i}_0$  as $\emph{red}$.
Moreover, if $\textsc{i}\imsucct{t_i}\textsc{i}'$ and $\textsc{i}$ is
viewed as red, then we also view $\textsc{i}'$ as red, and we also call the
	\emph{transition occurrence} $t_i$ \emph{red}.
Hence the red IDs are precisely the successors of $\textsc{i}_0$, and \vspace*{-1mm}
	\begin{center}
	$\sigma=t_1t_2\cdots t_m$ \vspace*{-1mm}
can be written as \vspace*{-1mm}
	$\sigma_0t'_1\sigma_1t'_2\sigma_2\cdots t'_r\sigma_r$ \vspace*{1mm}
	\end{center}
	where $\sigma_i\in T^*$ (for 	$i\in[0,r]$) and $t'_i$, $i\in[1,r]$, are precisely the red transition occurrences.

\eject

\noindent We put
\begin{equation}\label{eq:doublered}
\sigma'=\sigma_0t'_1t'_1\sigma_1t'_2t'_2\sigma_2\cdots
t'_rt'_r\sigma_r
\end{equation}
	(each red transition occurrence in $\sigma$ has been doubled),
and we finish the overall proof by the following claim, which entails
that $(N,M'_0)$ is non-live (as we have aimed to prove).

\smallskip \emph{Claim.}	 There is an execution $M'_0\step{\sigma'}M'$ (for $\sigma'$ in~(\ref{eq:doublered})),
and $\rst{(M')}{P\cruc}=\rst{(M\wit)}{P\cruc}$.

\medskip
The claim thus entails that all $t\in T\dead$ are dead in $(N,M')$, by
another use of Proposition~\ref{prop:ignoredead}.
To prove the claim,
we recall that $M'_0$ differs from $M_0$ just by an additional token
on $p_0$, and we construct an ID-valuation \vspace*{-1mm}
\begin{center}
$\idm'_0\step{\sigma_0t'_1t'_1\sigma_1t'_2t'_2\sigma_2\cdots
t'_rt'_r\sigma_r}\idm'$
\end{center}
by adjusting the above ID-valuation~(\ref{eq:IDvalorigexec})
so that each red ID gets a ``twin'' ID as follows:
\begin{itemize}	
\itemsep=0.95pt
	\item $\idm'_0$ arises from $\idm_0$	by equipping
the additional token on $p_0$ with a fresh ID $\textsc{i}'_0$ that
		is viewed as a \emph{twin of the red ID} $\textsc{i}_0$;
	\item 		for each $j\in[1,r]$, when the (red) $\textsc{i}$ is the
		source ID of the original (red)
		transition occurrence $t'_j$, then the twin ID of
		$\textsc{i}$ is the source of
		the added occurrence $t'_j$; moreover,
		each (red) destination ID of the original $t'_j$
gets a twin due to the added $t'_j$.
\end{itemize}
It is clear that $\idm'$ differs from $\idm_m$ only so that each red
ID in $\idm_m(p)$ has an additional twin in $\idm'(p)$ (for each $p\in
P$).
Since all successors of the red ID $\textsc{i}_0$ in $\idm_m$
are on
the don't care places (i.e.,
outside the set $P\cruc$ of crucial places),
all successors of its twin ID $\textsc{i}'_0$
are in
$\idm'$ on the don't care places as well; this entails that
$\rst{(M')}{P\cruc}=\rst{(M\wit)}{P\cruc}$.
\end{proof}

\begin{remark}
We note that the main idea
of the proof of the next proposition (Proposition~\ref{prop:liveupper})
is related to the proof idea
of Replacement Lemma in~\cite{DBLP:conf/concur/RaskinWE20}, though
it is presented in a different technical framework.
\end{remark}

\begin{proposition}\label{prop:liveupper}
Let $N=(P,T,F)$ be an ord-BIMO net, and let $M_0,M'_0$ be two markings of $N$ that
coincide except of one place $p_0$ for which we have
	$2|P|-1\leq M'_0(p_0)=M_0(p_0)-1$. If $(N,M_0)$ is non-live
	then
	$(N,M'_0)$ is
	non-live as well.
\end{proposition}

\begin{proof}
Let $N,M_0,M'_0,p_0$ be as in the statement, and let us assume
	that $(N,M_0)$ is non-live; we aim to show that $(N,M'_0)$ is
	non-live as well. (Now $M'_0(p_0)$ arises from
	$M_0(p_0)$ by removing a token, while in
	Proposition~\ref{prop:nonlivepmar} one token was added.)

\medskip
We fix
$M_0\step{\sigma}M\wit$, $P\cruc\subseteq P$,
	and $T\dead\subseteq T$ guaranteed by
	Lemma~\ref{lem:crucialingred}
	(in particular we recall that $|\rst{(M\wit)}{P\cruc}|<|P\cruc|$).
Similarly as in the proof of Proposition~\ref{prop:nonlivepmar},
we aim to show an execution
$M'_0\step{\sigma'}M'$
where $\rst{(M')}{P\cruc}=\rst{(M\wit)}{P\cruc}$, by which the proof
	will be finished.
We assume that $M_0\step{\sigma}M\wit$ is of the
form
\begin{equation*}
M_0\step{t_1}M_1\step{t_2}M_2\cdots\step{t_m}M_m=M\wit,
\end{equation*}
			and for this execution
we fix an ID-valuation
		\begin{equation}\label{eq:IDvalorigexectwo}
\idm_0\step{t_1}\idm_1\step{t_2}\idm_2\cdots\step{t_m}\idm_m.
		\end{equation}

	We view an ID $\textsc{i}$ in $\idm_0(p)$ as
			\emph{black}
if			there is a place $p'\in P\cruc$ and some
			$\textsc{i}'\in\idm_m(p')$ such that
			$\textsc{i}\imsucc^*\textsc{i}'$; hence there
			are at most $|P\cruc|-1$ black IDs in the set
			$\bigcup_{p\in P}\idm_0(p)$.
Now $M_0(p_0)\geq 2|P|$,
		and we thus have at least $|P|{+}1$ non-black IDs in
		$\idm_0(p_0)$
		(since $2|P|-(|P\cruc|-1)\geq |P|{+}1$). We view all non-black IDs in $\idm_0(p_0)$
		as \emph{red}, and all their successors
		(in~(\ref{eq:IDvalorigexectwo})) are viewed as red as well;
this also determines which \emph{transition occurrences} $t_i$
	in~(\ref{eq:IDvalorigexectwo}) are \emph{red}. We note that
	each red transition occurrence
		changes the distribution of red IDs while not
		affecting non-red IDs; on the other hand, each non-red
		transition occurrence does not affect red IDs (though
		it might need red IDs on its observation places).

\medskip
We write 	$\sigma=t_1t_2\cdots t_m$ 	as
	$\sigma_0t'_1\sigma_1t'_2\sigma_2\cdots t'_r\sigma_r$
		where 	$t'_i$, $i\in[1,r]$, are precisely the red
	transition occurrences in~(\ref{eq:IDvalorigexectwo}), and we
	also 	present~(\ref{eq:IDvalorigexectwo}) as
\begin{equation}\label{eq:IDvalorigexecthree}
\idmr_0\step{\sigma_0t'_1}\idmr_1\step{\sigma_1t'_2}\idmr_2\cdots
			\step{\sigma_{r-1}t'_r}\idmr_r
			\step{\sigma_r}\idmr_{r+1}
\end{equation}
where $\idmr_0=\idm_0$ and $\idmr_{r+1}=\idm_m$.
(The superscript $\textsc{r}$ in $\idmr_i$
is just a symbol that might be viewed as referring to ``oRiginal''
id-markings, which differ
from the constructed id-markings $\idm'_i$ in~(\ref{eq:IDvalorigexecfour}).)
We observe that $\idmr_{r}$ and $\idmr_{r+1}$ have the same
distribution of red IDs, since the segment $\idmr_{r}\step{\sigma_r}\idmr_{r+1}$
does not contain any red transition occurrence; in particular,
in both  $\idmr_{r}$ and $\idmr_{r+1}$
there
are no red IDs on the ``crucial'' places, i.e.\ on the places from $P\cruc$.

\medskip
We recall that we aim to construct a suitable execution
$M'_0\step{\sigma'}M'$; we will have
\begin{equation}\label{eq:sigmaprime}
	\sigma'=\sigma_0(t'_1)^{n_1}\sigma_1(t'_2)^{n_2}\sigma_2\cdots
	(t'_r)^{n_r}\sigma_r
\end{equation}
for certain multiplicities $n_1,n_2,\dots,n_r$ (also allowing $n_i=0$).
In the proof of Proposition~\ref{prop:nonlivepmar} we had $n_i=2$ for all $i\in[1,r]$ but here
(when $M'_0(p_0)=M_0(p_0)-1$)
the situation is more complicated.
The idea is that we aim to modify~(\ref{eq:IDvalorigexecthree})
so that we get
\begin{equation}\label{eq:IDvalorigexecfour}
	\idm'_0\step{\sigma_0(t'_1)^{n_1}}\idm'_1\step{\sigma_1(t'_2)^{n_2}}\idm'_2\cdots
			\step{\sigma_{r-1}(t'_r)^{n_r}}\idm'_r
			\step{\sigma_r}\idm'_{r+1}
\end{equation}
where $\idm'_0$ arises from $\idm_0$ (i.e., from $\idmr_0$) by
removing one red ID from $p_0$
(hence at least $|P|$ red IDs remain in $\idm'_0(p_0)$),
and where the red transition occurrences are used so
that~(\ref{eq:IDvalorigexecfour})
is an ID-valuation of
a desired execution $M'_0\step{\sigma'}M'$.
In~(\ref{eq:IDvalorigexecfour}) we will also view as red precisely
those IDs that are successors of the red IDs in $\idm'_0(p_0)$.
We will construct~(\ref{eq:IDvalorigexecfour}) so that the red
transition occurrences (changing just the distribution of red IDs)
will be precisely those in the segments $(t'_i)^{n_i}$ (for all $i\in[1,r]$);
each non-red transition occurrence, in a~segment $\sigma_i$
($i\in[0,r]$),
will cause the same ID-change as its corresponding non-red
transition occurrence in~(\ref{eq:IDvalorigexecthree}).
More concretely, we aim to choose the multiplicities $n_i$
so that $\sigma'$ (defined in~(\ref{eq:sigmaprime})) is performable from
$M'_0$ and allows us
to construct a~respective
ID-valuation~(\ref{eq:IDvalorigexecfour}) so that
the following conditions hold for all $i\in[0,r{+}1]$:
\begin{enumerate}
	\item the id-marking $\idm'_i$
		in~(\ref{eq:IDvalorigexecfour}) coincides with
		$\idmr_i$ in~(\ref{eq:IDvalorigexecthree}) when the red
		IDs are ignored;
\item	 for each $p\in P$: there is at least one red ID in
		$\idm'_i(p)$ in~(\ref{eq:IDvalorigexecfour}) if, and
		only if,
 in~(\ref{eq:IDvalorigexecthree}) there is a~red ID in some of the sets
	$\idmr_0(p),\idmr_1(p),\dots,\idmr_i(p)$ and
	there is a red ID in some of the sets
				$\idmr_i(p),
		\idmr_{i+1}(p),\dots,
		\idmr_{r+1}(p)$.
\end{enumerate}		
Hence in $\idm'_i(p)$ in~(\ref{eq:IDvalorigexecfour}) we aim to
remember, by a presence of a red ID, if in the corresponding situation
in~(\ref{eq:IDvalorigexecthree})
there has been a red ID on $p$ in the past or at present;
but there should be no red ID in $\idm'_i(p)$ if
in the corresponding situation
in~(\ref{eq:IDvalorigexecthree})
there is no red ID on $p$ at present and in the future.

\medskip
Conditions $1$ and $2$
clearly hold in the case $i=0$, by our definition of $\idm'_0$
(arising from $\idmr_0$). Moreover, if conditions $1$ and $2$
hold in the case $i=r{+}1$ (and $\sigma'$ is performable from $M'_0$),
then we are done, since in this case
$\rst{(\idm'_{r+1})}{P\cruc}=\rst{(\idmr_{r+1})}{P\cruc}$
(due to the fact that there are no red IDs in $\bigcup_{p\in
P\cruc}\idmr_{r+1}(p)$), which entails
$\rst{(M')}{P\cruc}=\rst{(M\wit)}{P\cruc}$.

\medskip
Informally speaking, at most $|P|$ red IDs will suffice to serve as
the required ``memory''; we will create them from the successors
of at most $|P|$ red IDs in $\idm'_0(p_0)$ (recall that $\idm'_0(p_0)$ contains at least $|P|$ red
IDs). But we should also take care of the required ``cleaning''
(to get $\rst{(M')}{P\cruc}=\rst{(M\wit)}{P\cruc}$); e.g.,
if $p_0\in P\cruc$, then we have to remove all red IDs from $p_0$.

\medskip
To define the multiplicities $n_i$
in~(\ref{eq:IDvalorigexecfour}) rigorously, we first introduce a few
technical notions.
Referring to~(\ref{eq:IDvalorigexecthree}), for $i\in[0,r{+}1]$ we put
\begin{center}
		$\textsc{R}_i=\{p\in P\mid \idmr_i(p)$ contains a red ID$\}$.
\end{center}
As we observed after defining~(\ref{eq:IDvalorigexecthree}),
we have $\textsc{R}_r=\textsc{R}_{r+1}$, and $\textsc{R}_r\cap P\cruc=\emptyset$.
We define ``first-red'' sets and ``last-red'' sets for
		all $i\in[0,r{+}1]$:
\begin{itemize}
\itemsep=0.9pt
	\item	 $\textsc{FR}_i=\{p\in P\mid  p\in \textsc{R}_i$ and $p\not \in \textsc{R}_0\cup
			\textsc{R}_1\cdots\cup  \textsc{R}_{i-1}\}$;
	\item $\textsc{LR}_i=\{p\in P\mid  p\in \textsc{R}_i$ and $p\not \in
				\textsc{R}_{i+1}\cup  \textsc{R}_{i+2}\cdots\cup
			\textsc{R}_{r+1}\}$.
\end{itemize}
We note that $\textsc{FR}_0=\{p_0\}$,
the sets $\textsc{FR}_i$ ($i\in[0,r{+}1]$) are pairwise disjoint (hence
$|\bigcup_{i\in[0,r{+}1]}\textsc{FR}_i|\leq |P|$),
and for 	$i\in[1,r]$ each place in $\textsc{FR}_i$ is a destination place of $t'_i$ (which
holds trivially in the case $\textsc{FR}_i=\emptyset$).
For 	$i\in[0,r{-}1]$, if the set $\textsc{LR}_i$ is nonempty
 then it is a singleton 	consisting of the source place of $t'_{i+1}$.
For 	$i\in[1,r]$, we say that $t'_i$ 	in~(\ref{eq:IDvalorigexecthree}) is
		\begin{itemize}
\itemsep=0.9pt
			\item a \emph{first-red-destination 	transition-occurrence}, a \emph{first-rdto}
				for short, if 	$\textsc{FR}_i\neq\emptyset$;
	\item a \emph{last-red-source transition-occurrence},
 a \emph{last-rsto} for short,  if  $\textsc{LR}_{i-1}\neq\emptyset$.
	\end{itemize}
Since  the sets $\textsc{FR}_i$
are pairwise disjoint and
$\textsc{FR}_0=\{p_0\}$, there are at most $|P|{-}1$ first-rdtos.
We define a \emph{causality relation $\causr$ on
		the set of first-rdtos}:
\begin{center}
		$t'_i\causr t'_j$ if the source
		place of $t'_j$ belongs to $\textsc{FR}_i$ (and thus is
		a destination place of $t'_i$).
\end{center}

Hence each first-rdto $t'_i$ either has $p_0$ as the source place,
in which case there is no $j$ such that $t'_j\causr t'_i$,
or there is
precisely one $j$ ($j<i$) such that $t'_j\causr t'_i$.
Viewing the relation $\causr$ as a directed graph, it has
the form of a forest (a set of directed disjoint trees).
By $\causr^*$ we denote the reflexive and transitive closure of
$\causr$, and for each first-rdto $t'_i$ we put
\begin{center}
	$\postcr{t'_i}=\{j\mid t'_i\causr^* t'_j\}$.
\end{center}	
\eject

We define the numbers $n_i$, $i\in[1,r]$, in~(\ref{eq:IDvalorigexecfour})
as follows, depending on $t'_i$ in~(\ref{eq:IDvalorigexecthree}):

\begin{itemize}
\item if $t'_i$ is a first-rdto but not a last-rsto, then
		$n_i=|\postcr{t'_i}|$;
	\item if $t'_i$ is a last-rsto, with the source place $p_s$, then
		$n_i=|\{\textsc{i}\mid \textsc{i}$ is a red ID in
			$\idm'_{i-1}(p_s)\}|$;
	\item if $t'_i$ is neither a first-rdto nor a last-rsto, then $n_i=0$.
\end{itemize}
We note that we have not excluded that some $t'_i$ in~(\ref{eq:IDvalorigexecthree})
is both a first-rdto
(some of its destination places gets a red ID for the first time) and a last-rsto
(its source place gets rid of red IDs for the rest of the
execution). We also note that $n_i\geq 1$ for each first-rdto $t'_i$.

\medskip
To show the validity of the above conditions $1$ and $2$, it is useful
to add the following condition, for the cases $i\in[1,r]$:

\begin{enumerate}
		\setcounter{enumi}{2}
	\item
		If $t'_i$ is a first-rdto
		(in~(\ref{eq:IDvalorigexecthree}))
		and $p_s$ is its source place, then
\begin{center}
the set $\idm'_{i-1}(p_s)$
	(in~(\ref{eq:IDvalorigexecfour}))
	contains at least $1+|\postcr{t'_{i}}|$ red IDs.
\end{center}
\end{enumerate}

\noindent This condition should thus hold also in the case when
$t'_i$ is both a first-rdto and a last-rsto.

\medskip
Now we show that the chosen $n_i$ indeed fulfill our goals, i.e.,
$\sigma'$ (in~(\ref{eq:sigmaprime})) is performable from $M'_0$ and
we construct a respective ID-valuation~(\ref{eq:IDvalorigexecfour}) so
that conditions $1-3$ are satisfied for all $i\in[0,r{+}1]$.

\medskip
We first assume that condition $3$ holds for all $i\in[1,r]$,
and under this assumption we show that $1$ and $2$ are then satisfied for all
$i=0,1,\dots,r{+}1$;
we use an induction on $i$. We have already noted that
 $1$ and $2$ are satisfied for $i=0$ (where condition $3$
does not apply).
In the induction step we fix $j\in[0,r]$ and assume
that $\sigma_0(t'_1)^{n_1}\sigma_1(t'_2)^{n_2}\cdots
	\sigma_{j-1}(t'_j)^{n_j}$ is performable from $M'_0$ and
	conditions $1$ and $2$ are satisfied for $i=j$, in the so far constructed ID-valuation
$\idm'_0\step{\sigma_0(t'_1)^{n_1}}\idm'_1\cdots\step{\sigma_{j-1}(t'_j)^{n_j}}\idm'_j$.
To extend the validity to the case $i={j+1}$, we
first define the segment $\idm'_j\step{\sigma_j}\idm'_{\textsc{a}}$ in~(\ref{eq:IDvalorigexecfour})
by mimicking the segment $\idmr_j\step{\sigma_j}\idmr_{\textsc{a}}$
in~(\ref{eq:IDvalorigexecthree}); the substrict $\textsc{a}$ is just
auxiliary. Both segments thus perform the same (non-red) ID-changes,
which is possible due to conditions $1$ and $2$ for the case $i=j$;
in particular we note that for each red ID used in observation places in the segment
$\idmr_j\step{\sigma_j}\idmr_{\textsc{a}}$ there is a respective red
ID that can be used in $\idm'_j\step{\sigma_j}\idm'_{\textsc{a}}$.
In the case $j=r$ we have $\idm'_{\textsc{a}}=\idm'_{r+1}$, and $1$
and $2$ obviously hold for $i=j{+}1$ as well.
If $j<r$, then  conditions $1$ and $2$ for $i=j$
and condition $3$ for $i=j{+}1$ (which is so far just assumed)
guarantee that we can add a~segment $\idm'_\textsc{a}\step{(t'_{j+1})^{n_{j+1}}}\idm'_{j+1}$ in which all
transition occurrences are red. Moreover, both conditions $1$ and $2$
are satisfied for $i=j{+}1$:

\begin{itemize}
\itemsep=0.98pt
	\item Condition $1$ follows from the fact that the segment
		$\idm'_j\step{\sigma_j}\idm'_{\textsc{a}}$
		mimics the segment 	$\idmr_j\step{\sigma_j}\idmr_{\textsc{a}}$ (making
		precisely the same ID-changes for non-red IDs), while
		the segments  $\idmr_{\textsc{a}}\step{t'_{j+1}}\idmr_{j+1}$ and
		$\idm'_\textsc{a}\step{(t'_{j+1})^{n_{j+1}}}\idm'_{j+1}$
		affect only red IDs.
	\item	If $t'_{j+1}$ is a first-rdto but not a last-rsto
		(in~(\ref{eq:IDvalorigexecthree})), 	then $n_{j+1}\geq 1$ and performing
		$(t'_{j+1})^{n_{j+1}}$ in~(\ref{eq:IDvalorigexecfour})
		leaves at least one red ID in  the source place of
		$t'_{j+1}$; hence condition $2$ is surely kept.
	\item 	If $t'_{j+1}$ is a first-rdto and a last-rsto, 	
		then $n_{j+1}\geq 1$ and performing
		$(t'_{j+1})^{n_{j+1}}$ in~(\ref{eq:IDvalorigexecfour})
		leaves no red ID in  the source place of
		$t'_{j+1}$; hence condition $2$ is kept as well.
	\item If $t'_{j+1}$ is not a first-rdto (while it might be or not be
		a last-rsto), keeping 	condition $2$ is also obvious.		
\end{itemize}		
It remains to deal with condition $3$, which should hold for all
$i\in[1,r]$ for which $t'_i$ are first-rdtos; we recall that there are
at most $|P|-1$ first-rdtos.
Now we use an induction based on the relation $\causr$ (defined on the
set of first-rdtos); we recall
that $\causr$ can be naturally viewed as a set of disjoint directed
trees.
We fix the source place $p$ of some first-rdto $t'_i$ and define
\begin{center}
	$\textsc{S}_p=\{i\in[1,r]\mid t'_i$ is a
first-rdto in~(\ref{eq:IDvalorigexecthree}) for which $p$ is the
	source place$\}$.
\end{center}
Let $\textsc{S}_p=\{i_1,i_2,\dots,i_k\}$ where $i_1< i_2\cdots<
i_k$. We note that
if there exists a last-rsto $t'_j$ whose source
place is $p$, then $i_k\leq j$.
We have either $p=p_0$,
or there is a unique $i_0$ ($i_0<i_1$) such that $t'_{i_0}\causr t'_i$ for all
$i\in\textsc{S}_p$ (in which case $p$ is a destination place of
$t'_{i_0}$, belonging to $\textsc{FR}_{i_0}$). We deal with these two cases
below (implicitly using the fact that the definition of $\causr$ entails that the sets
$\postcr{t'_{i_1}}$, $\postcr{t'_{i_2}}$,
$\dots$, $\postcr{t'_{i_k}}$ are pairwise disjoint):

\begin{itemize}		
	\item
		$p=p_0$ (basis of our induction)

		Since $\sum_{i\in\textsc{S}_{p_0}}|\postcr{t'_{i}}|\leq|P|-1$ and there are at
		least $|P|$ red IDs in $\idm'_0(p_0)$, condition $3$
		clearly holds for all $i\in\textsc{S}_{p_0}$.
	\item
 $t'_{i_0}\causr t'_i$ for all
		$i\in\textsc{S}_p$ (induction step, assuming that
		condition $3$ holds for $i_0$)

		We note that $\idm'_{i_0}(p)$ contains $n_{i_0}$ red
		IDs, hence at least $|\postcr{t'_{i_0}}|$
		red IDs. (Here we use the inductive assumption for
		$i_0$, which guaranteed that in the segment
		$(t'_{i_0})^{n_{i_0}}$ of~(\ref{eq:IDvalorigexecfour})
		we could indeed make all transition occurrences red.)
Since
\begin{equation*}
	n_{i_0}\geq|\postcr{t'_{i_0}}|=1+\sum_{i,t'_{i_0}\causr t'_{i}}
	|\postcr{t'_{i}}|\geq 1+\sum_{i\in\textsc{S}_p}
	|\postcr{t'_{i}}|, \vspace*{-2mm}
\end{equation*}
condition $3$ holds for all $i\in\textsc{S}_p=\{i_1,i_2,\dots,i_k\}$.
\end{itemize}
\end{proof}

\subsection{An ord-IMO net is structurally live iff there is a live  $\{0,1\}$-marking}\label{sec:oIMOcase}

In this section we prove the following fact for ord-IMO nets
(thus proving another part of Theorem~\ref{thm:upperboundstable}
that refers to Table~\ref{tab:res});
Section~\ref{sec:counterBIO} shows that this does not hold for ord-BIO
nets.

\begin{lemma}[For structural liveness of ord-IMO nets the $\{0,1\}$ markings are
	decisive]\label{lem:IMOsmall}
	An \oIMO net $N=(P,T,F)$
	is structurally live if{f} there is a marking $M_{0}:P
	\rightarrow \{0,1\}$ such that $(N,M_{0})$ is live.
\end{lemma}	
\eject

\begin{proof}
The ``if'' direction is trivial, so we now deal with the ``only-if''
	direction.
For the sake of contradiction we assume a fixed \oIMO net $N=(P,T,F)$
	for which a fixed marking $M_0$ is live but all $M'_{0}:P\rightarrow \{0,1\}$
	are non-live.

\medskip
	By Proposition~\ref{prop:liveentailscmsc} (and the fact that
	all $M\in\rset{M_0}$ are live if $M_0$ is live) we can assume that
	$M_0$ is optimal. Since $M_0$ is thus self-coverable and the
	\oIMO net $N$
	is conservative, there must be some full
	$\sigma=t_1t_2\cdots t_m$ (in which each transition from $T$
	occurs at least once) such that
	\begin{equation}\label{eq:basicfull}
		M_0\step{\sigma}M_0, \textnormal{ i.e. }M_0\step{t_1t_2\cdots t_m}M_0.
	\end{equation}
This easily shows that all strongly connected components of $\prlx{N}$
	are pairwise isolated (there is no edge from one scc to another
	scc). Hence for each $M$ and all $M'\in\rset{M}$ we have
	$|\rst{M'}{P_C}|=|\rst{M}{P_C}|$ for all components $C$ of
	$\prlx{N}$. (The number of tokens in each component is stable.)

\medskip
Since $M_0$ is optimal, by Lemma~\ref{lem:spread}
	for each component $C$ that is rich at $M_0$ we have
	$M_0(p)\geq 1$ for all $p\in P_C$, and for each component $C$
	that is poor at $M_0$ (hence $|\rst{(M_0)}{P_C}|<|P_C|$) we
	have $M_0(p)\in\{0,1\}$ for all $p\in P_C$.
Let $M'_0$ be the ``carrier-marking'' related to $M_0$, i.e., for each
$p\in P$ we have
	\begin{center}
		$M'_0(p)=0$ if $M_0(p)=0$, and 	$M'_0(p)=1$ if
		$M_0(p)\geq 1$.
	\end{center}
We say that $M'_0$ arises from $M_0$ by removing ``superfluous
	tokens'' (from rich components), and note 	
that $(N,M'_0)$ has the same rich components as $(N,M_0)$.
We also recall that $M'_0$ is non-live by our assumption.
Lemma~\ref{lem:crucialingred} and its proof thus yield
	\begin{center}
	$M'_0\step{\sigma'}M\wit$, $P\cruc$, $T\dead$
	\end{center}
		where $M\wit$ is a DL-marking 		for which the transitions from $T\dead$
are dead and the transitions from $T_L=T\smallsetminus T\dead$ are 	live.
Moreover, due to the proof of Lemma~\ref{lem:crucialingred}
we can choose $M\wit$ so that
	$P\cruc$ is the union of the sets $P_C$ for the components $C$
	of $\prlx{\rst{N}{(P,T_L)}}$ that are poor in
$(\rst{N}{(P,T_L)},M\wit)$.

\medskip
	We now aim to show that there is $\bar{M}_0\in\rset{M_0}$ such
	that $\bar{M}_0\geq M'_0$ and
	$\rst{(\bar{M}_0)}{P\cruc}=\rst{(M'_0)}{P\cruc}$; then the
	execution $M_0\step{*}\bar{M}_0\step{\sigma'}M'\wit$ contradicts the
	assumption that $M_0$ is live (since
	$\rst{(M'\wit)}{P\cruc}=\rst{(M\wit)}{P\cruc}$ and transitions
	from $T\dead$ are thus dead at $M'\wit$). The proof
	will thus be finished.

\medskip
	We first note that for each component $C'$ of
	$\prlx{\rst{N}{(P,T_L)}}$ we have that
	$P_{C'}$	is a subset of  $P_{C}$ for some component $C$ of
	$\prlx{N}$; in other words, each component of $\prlx{N}$ is
	partitioned into some components of
	$\prlx{\rst{N}{(P,T_L)}}$. Since for each component $C$ of $\prlx{N}$
	that is rich at (the above carrier-marking) $M'_0$
we have $|\rst{M}{P_C}|=|\rst{(M'_0)}{P_C}|=|P_C|$
for all $M\in\rset{M'_0}$, the partition of $C$ into the components
of $\prlx{\rst{N}{(P,T_L)}}$ must contain at least one component $C'$
that is rich in $(\rst{N}{(P,T_L)},M\wit)$.
Hence for each component $C$ of $\prlx{N}$ that is rich at $M'_0$
	there is a~place $p_\textsc{nc}\in P_C\smallsetminus P\cruc$
	(hence $p_\textsc{nc}$ is a don't care place in
	$(\rst{N}{(P,T_L)},M\wit)$).

\medskip
We recall that the rich components in $(N,M'_0)$ coincide with the rich
components in $(N,M_0)$. Our goal to show $M_0\step{*}\bar{M}_0$ where
	$\rst{(\bar{M}_0)}{P\cruc}=\rst{(M'_0)}{P\cruc}$ will be realized
	by moving all superfluous tokens in each component $C$ that is
	rich in $(N,M_0)$ onto a chosen place $p_\textsc{nc}\in
	P_C\smallsetminus P\cruc$.
It is clearly sufficient to show how to modify the execution~(\ref{eq:basicfull}) so that
$M_0$ is transformed just so that
one superfluous token from some $p\in P\cruc$ is moved outside $P\cruc$, onto a particular
$p_\textsc{nc}$. Such a process can be repeated until we get the
desired $\bar{M}_0$, which finishes the proof.

\medskip
Hence we fix some $p\in P\cruc$ where $M_0(p)\geq 2$ (and thus
$M_0(p)>M'_0(p)$); by our choice of
$M_0$ we have
$p\in P_C$ for a component $C$  that is rich
	in $(N,M_0)$. We fix some
	$p_\textsc{nc}\in P_C\smallsetminus P\cruc$ (whose existence
	has been discussed).
There must be a path from $p$ to $p_\textsc{nc}$ (in $\prlx{N}$); let
$t'_1t'_2\cdots t'_r$ be the sequence of transitions on this path.
We consider the execution
\begin{equation}\label{eq:movingsuperfluous}
M_0\step{\sigma_1}M_1\step{\sigma_2}M_2\cdots\step{\sigma_r}M_r
\end{equation}
where $\sigma_i$ ($i\in[1,r]$) arises from $\sigma$
in~(\ref{eq:basicfull}) by omitting all transitions contained in rich
components, except of one occurrence of $t'_i$
(hence $t'_1$ occurs once in $\sigma_1$, $t'_2$ occurs once in
$\sigma_2$, etc.).
Since all places in rich components are marked in $M_0$,
it is easy to check that~(\ref{eq:movingsuperfluous}) is a valid execution,
and that $M_r$ coincides with $M_0$ except that
$M_r(p)=M_0(p)-1$ (hence $M_r(p)\geq 1$) and
$M_r(p_\textsc{nc})=M_0(p_\textsc{nc})+1$.
Hence  the carrier-markings of $M_r$ and
$M_0$ are the same, namely
$M'_0$, but the amount of superfluous tokens in $M_r$ is less than in
$M_0$.
\end{proof}

\subsection{A structurally live ord-BIO net in which all $\{0,1\}$-markings
are non-live}\label{sec:counterBIO}

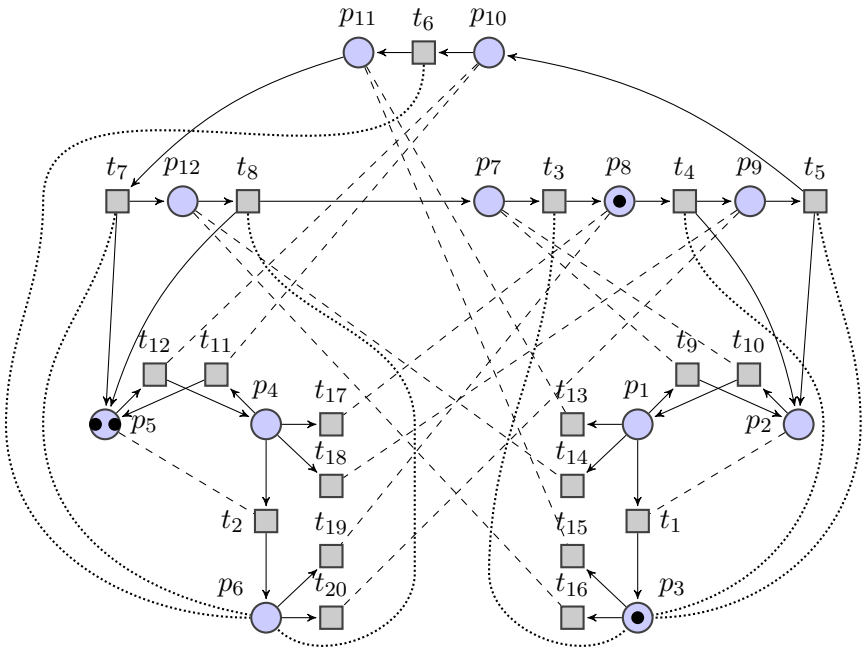
\begin{figure}[!b]
	\centering
\scalebox{0.98}{
        \input{figures/example-bio.tex} }\vspace*{-2mm}
        \caption{A structurally live   \oBIO net $N$ in which all
        $M:P \longrightarrow \{0,1\}$ are non-live.} \label{fig:bioce}
\end{figure}

Here we show that Lemma~\ref{lem:IMOsmall}
cannot be extended to ord-BIO nets, by providing a concrete example.
Roughly speaking, if we wanted to mimic the proof
of Lemma~\ref{lem:IMOsmall} in the case of BIO nets,
a problem would arise at moving superfluous tokens to don't care
places: such moving can create further superfluous tokens due to
branching transitions.

\medskip
Our example net $N$ is depicted in Figure~\ref{fig:bioce};
for lucidity, instead of
drawing two observation edges $(p,t)$ and $(t,p)$
we draw just one, dotted or dashed, edge between $p$ and $t$ without
end-arrows.
The respective
properties of $N$ are formulated by the next two propositions.

\begin{proposition}\label{prop:counterBIOfirst}
	The net $N=(P,T,F)$ in Figure~\ref{fig:bioce} is structurally	live.
\end{proposition}

\begin{proof}
	Let $M_0$ be the marking of $N$ depicted in
	Figure~\ref{fig:bioce}, i.e.,
	 $M_0(p_3) = M_0(p_8)=1$,
	$M_0(p_5) = 2$, and $M_0(p)=0$
	for all 	$p \in P\smallsetminus \{p_3,p_5,p_8\}$.
The set of places  $\{p_7,p_8,p_9,p_{10}, p_{11},p_{12}\}$
	can be viewed as a ``control part'' in which
 precisely one token is moving; i.e.,
	for all $M \in \rset{M_0}$ we have
    \begin{enumerate}
	    \item $\sum_{i\in[7,12]} M(p_i)=1$.
    \end{enumerate}
The first three of ``control transitions'' $t_3,t_4,t_5,t_6,t_7,t_8$
	require an observation token
in $p_3$, while the last three require an observation token
in $p_6$; we also note that each of $t_4$ and $t_5$ adds a fresh token
to $p_2$, while each of $t_7$ and $t_8$ adds a fresh token
to $p_5$.
We now add the following
    conditions that are satisfied for all $M \in \rset{M_0}$:
    \begin{enumerate}
    \itemsep=0.95pt
	    \setcounter{enumi}{1}
        \item if  $\sum_{i\in[7,9]} M(p_i)=1$, then $M(p_5) \geq 2$;
        \item if $\sum_{i\in[10,12]} M(p_i)=1$, then $M(p_2) \geq 2$;
        \item if $M(p_7) = 1$, then $M(p_1)+M(p_2) \geq 2$ or $M(p_3)
		\geq 1$;
        \item if $M(p_{10}) = 1$, then $M(p_4)+M(p_5) \geq 2$ or
		$M(p_6) \geq 1$;
        \item if $M(p_8)+M(p_9) = 1$, then $M(p_3) \geq 1$;
        \item if $M(p_{11})+M(p_{12}) = 1$, then $M(p_6) \geq 1$;
        \item if $M(p_9)=1$, then $M(p_2) \geq 1$;
        \item if $M(p_{12})=1$, then $M(p_5) \geq 1$.
    \end{enumerate}
We easily check that the conditions $1-9$ are
	satisfied for $M_0$ (i.e., for $M=M_0$).
Now let $M_1$ satisfy $1-9$, and let
	$M_1\step{t} M_2$ (for some $t \in T$).
It is a routine to verify that $M_2$ satisfies $1-9$ as well:
	E.g., if $t=t_3$, then $M_1(p_7)=1$, $M_1(p_3)=M_2(p_3)\geq 1$, and
	$M_2(p_8)=1$; $M_2$ thus obviously satisfies the conditions $1-9$
	(in particular $6$).

\medskip
In fact, we have shown that if $M$ satisfies $1-9$, then each
$M'\in\rset{M}$ satisfies $1-9$ as well. Now we show that if $M$ satisfies $1-9$, then no transition is
	dead at $M$; this entails that each $M$ satisfying $1-9$ 	(which includes $M_0$) is live.

We fix some  $M$ satisfying $1-9$. By condition $1$,
there is exactly one control 	transition $t\in\{t_3,t_4,t_5,t_6,t_7,t_8\}$
	whose source place is marked (by one token) at $M$; we perform a case analysis:
	\begin{itemize}	
		\item  If $t=t_3$ (hence $M(p_7)=1$), then by condition
			$4$ we have either $M(p_3)\geq 1$, or we can
			mark $p_3$ by transition $t_1$ that might need
			to be preceded by $t_9$ or $t_{10}$; then
			$t_3$ can be executed.
		\item  If $t=t_6$ (hence $M(p_{10})=1$), then by condition
			$5$ we have either $M(p_6)\geq 1$, or we can
			mark $p_6$ by transition $t_2$ that might need
			to be preceded by $t_{11}$ or $t_{12}$; then
			$t_6$ can be executed.
		\item If $t \in \{t_4,t_5\}$, then $t$ can be executed by condition $6$.
\item If $t \in \{t_7,t_8\}$, then $t$ can be executed by condition $7$.
	\end{itemize}
From $M$ we can thus perform all control transitions so that both
$p_3$ and $p_6$ are marked afterwards. Then we can be further
executing just the control transitions, which is increasing the number
of tokens on $p_2$ and $p_5$. This makes clear that all transitions
(including the ``token-consuming transitions''  $t_{13}, t_{14},\dots,t_{20}$)
can become enabled when we start from $M$.
\end{proof}

\begin{proposition}\label{prop:counterBIOsecond}
For the net $N=(P,T,F)$ in Figure~\ref{fig:bioce}, each marking
	$M:P\rightarrow\{0,1\}$ is non-live.
\end{proposition}

\begin{proof}
We first show that any marking $M$ satisfying one of the following
	two conditions
	\begin{enumerate}[a)]
   \itemsep=0.96pt
		\item $M(p_1)+M(p_2)\leq 1$ and	$M(p_3)=0$,
		\item $M(p_4)+M(p_5)\leq 1$ and $M(p_6)=0$
	\end{enumerate}
 is non-live. In the case a), $t_1,t_3,t_4,t_5$ are dead at $M$, since
	each of  $t_3,t_4,t_5$  needs to its enabling that $t_1$ is performed earlier,
	while $t_1$ needs that $t_4$ or $t_5$ is performed earlier.
	The case b) is analogous, here $t_2,t_6,t_7,t_8$ are dead at $M$.

\medskip
Now we fix a marking $M_0:P \longrightarrow \{0,1\}$
and show that $M_0$ is non-live,
	by analysing the
	following cases C1 and C2.
\begin{enumerate}
    \item[C1] $M_0(p_8)+M_0(p_9)>0$ or $M_0(p_{11})+M_0(p_{12})>0$.

If $M_0(p_{11})+M_0(p_{12})>0$, then we have
$M_0\step{\sigma}M$ where $\sigma$ contains
		no other transitions than
		$t_{13},t_{14},t_{15},t_{16}$,
and
		$M(p_1)=M(p_3)=0$, while $M(p_2)=M_0(p_2)\leq 1$.
		Hence $M$ is non-live by the above case a), which
		entails that $M_0$ is non-live as well.
		If $M_0(p_{8})+M_0(p_{9})>0$, then we get the case b)
		analogously.

   \item[C2] $M_0(p_8)+M_0(p_9)=0$ and $M_0(p_{11})+M_0(p_{12})=0$.

We assume that $M_0\step{\sigma t}M$ is a shortest execution where
		$t\in\{t_3,t_6\}$; if it does not exist, then $M_0$ is
		non-live.
For concreteness, we now assume that $t=t_6$, and consider the
		following two cases separately:
\begin{itemize}		
\item
$M_0(p_7)=0$.\\
	We observe that $M(p_2)\leq
		M_0(p_2)\leq 1$ (since $t_3$ and $t_6$ do not occur in
		$\sigma$, and thus $t_4,t_5,t_7,t_8$, and $t_9$ cannot occur in
		$\sigma$ either). Since $M(p_{11})=1$, we have
		$M\step{\sigma'}M'$ where
		$\sigma'=(t_{13})^{M(p_1)}(t_{15})^{M(p_3)}$ and
		$M'(p_1)=M'(p_3)=0$. Since $M'(p_2)\leq 1$,
		the above case a) shows that
		$M'$ is non-live, which entails that $M_0$ is non-live
		as well.
	\item
		$M_0(p_7)=1$.
\\
		Here $M(p_{11})=M(p_7)=1$ (since neither
		$t_3$ nor $t_8$ occurs in $\sigma$).
We have
$M\step{\sigma'}M'$ where
		$\sigma'=(t_{10})^{M(p_2)}(t_{13})^{M(p_1)+M(p_2)}(t_{15})^{M(p_3)}$
		 and
		$M'(p_1)=M'(p_2)=M'(p_3)=0$. Again, we apply the
		above case a) to $M'$, and deduce that $M_0$ is
		non-live.
\end{itemize}
		The case $t=t_3$ (instead of
		$t=t_6$) is analogous; here the above case b)
		applies.\vspace*{-1mm}
\end{enumerate}

\vspace*{-6mm}
\end{proof}

\section{Extension to non-ordinary nets}\label{sec:extensions}

In this section we finish proving Theorem~\ref{thm:upperboundstable}
that is captured by Table~\ref{tab:res} (in Section~\ref{subsec:results}).
This will be accomplished by proving the next lemma.

\begin{lemma}[Remaining results to fill
	Table~\ref{tab:res}]\label{lem:finalbounds}
Let $N = (P,T,F)$ be a BIMO net in which the maximum edge-weight is
	$w$. We have:
	\begin{enumerate}
 \itemsep=0.9pt
        \item $N$ is structurally live if{f} there is $M:P\rightarrow \{0,\dots,w \cdot \sizeof{P}\}$ such that $(N,M)$
	is live.
\item $(N,M)$ is live if{f}  $(N,M')$ is live whenever $M'(p) = M(p)$ for all $p\in P$ such that \\
			$\min\{M(p),M'(p)\}<2\cdot w \cdot  \sizeof{P}$.
\item If $N$ is an IMO net, then it is 	structurally live if{f} there is $M:P
	\rightarrow \{0,\dots,w\}$ such that $(N,M)$ is live.
    \end{enumerate}
\end{lemma}

We recall that the maximum edge-weight in IO nets is $2$, which
entails that the bounds for IMO nets in Table~\ref{tab:res} entail the
bounds for IO nets.
We prove Lemma~\ref{lem:finalbounds} by a simple construction
(illustrated in Figure~\ref{fig:bimoext}) that
gives an
ord-X net to a given X net, for X$\in\{$IO, IMO, BIO, BIMO$\}$,
 and by
recalling the results for ordinary nets.

\begin{figure}[!h]
\vspace*{-4mm}
    \centering
        \input{figures/example-extension.tex}\vspace*{-3mm}
    \caption{Transforming a BIMO transition  $\bimotrans{t}{p_1}{p_1}{p_1,p_1,p_2}$
	to an ord-BIMO transition $\bimotrans{t'}{p_\lara{1,1}}{p_\lara{1,2}}{p_\lara{1,1},p_\lara{1,3},p_\lara{2,1}}$.}
    \label{fig:bimoext}\vspace{-3mm}
\end{figure}

\paragraph{An ord-BIMO net $N'$ related to a given BIMO net $N$.}
Given a BIMO net $N = (P,T,F)$, we define the function $\wmax:P
\longrightarrow \Nat$ such that
\begin{center}
$\wmax(p) = \max(\{F(p,t) \mid t
\in T\} \cup \{F(t,p) \mid t \in T\})$.
\end{center}
(For instance, in Figure~\ref{fig:bimoext} we have $\wmax(p_1)=3$.)
\eject

%\medskip
From $N=(P,T,F)$, where $P=\{p_1,p_2,\dots,p_m\}$,
we create an ord-BIMO net
$N'=(P',T',F')$ as follows (cf.\ Figure~\ref{fig:bimoext}):
\begin{itemize}
	\item $P'$ arises from $P$ so that each place $p_i  \in P$ with $\wmax(p_i) \geq 1$ is replaced by new places
           $p_\lara{i,1},p_\lara{i,2},\dots,p_\lara{i,\wmax(p_i)}$.
\item 	$T'=\{t'\mid t\in T\} 	\cup\bigcup_{i\in[1,m]}\{t_{\lara{i,1}},t_{\lara{i,2}},\dots,t_{\lara{i,\wmax(p_i)}}\}$
	where
 \begin{itemize}
 \itemsep=0.96pt
			\item for each $i\in[1,m]$ we have  $\iotrans{t_\lara{i,j}}{p_\lara{i,j}}{}{p_\lara{i,j+1}}$ for all
				$j\in[1,\wmax(p_i){-}1]$, and    $\iotrans{t_\lara{i,\wmax(p_i)}}{p_\lara{i,\wmax(p_i)}}{}{p_\lara{i,1}}$;
    \item     each edge $(p_i,t)$ in $N$, with weight $F(p_i,t)$, gives
				rise to ordinary edges  $(p_\lara{i,1},t')$, $(p_\lara{i,2},t'),\dots,(p_\lara{i,F(p_i,t)},t')$ 	in $N'$;
    \item     each edge $(t,p_i)$ in $N$, with weight $F(t,p_i)$, gives
				rise to ordinary edges        $(t',p_\lara{i,1})$, $(t',p_\lara{i,2}),\dots,(t',p_\lara{i,F(p_i,t)})$ 	in $N'$.
\end{itemize}
\end{itemize}
It is easy to verify that the above ord-BIMO net $N'$ related to a BIMO net
$N$ is an ord-X net if $N$ is an X net, for X$\in\{$BIO, IMO,IO$\}$.

\medskip
We say that a marking $M$ of $N=(P,T,F)$, where $P=\{p_1,p_2,\dots,p_m\}$,
and a marking $M'$
of $N'=(P',T',F')$ are \emph{related}, which we
denote by
$M\approx M'$,
if for each
$p_i\in P$ with $\wmax(p_i)\geq 1$ we have
\begin{center}
$M(p)=\sum_{j\in[1,\wmax(p_i)]}M'(p_\lara{i,j})$.
\end{center}
It is straightforward to verify that $M_1\approx M_2$ entails:
\begin{itemize}
\itemsep=0.96pt
	\item 	if $M_1\step{t}M'_1$, then
 $M_2\step{\sigma t'}M'_2$ for some $\sigma$ consisting of occurrences of
		$t_\lara{i,j}$ ($i\in[1,m]$, $j\in[1,\wmax(p_i)]$),
		and  we have $M'_1\approx M'_2$;
	\item
		if $M_2\step{t_\lara{i,j}}M'_2$, then $M_1\approx
		M'_2$;
	\item
		if $M_2\step{t'}M'_2$, then $M_1\step{t}M'_1$ where
		$M'_1\approx M'_2$;
	\item
 $M_1$ is live in $N$ iff $M_2$ is live in $N'$.
\end{itemize}

The previous results for ordinary nets and the above construction
thus entail Lemma~\ref{lem:finalbounds}.

Moreover, Lemma~\ref{lem:crucialingred} can be generalized as follows:

\begin{lemma}[For BIMO nets, if $M_0$ is non-live, then there is
	a simple witness
	$M\wit\in\rset{M_0}$]\label{lem:crucialingredgener}
For a BIMO net $N=(P,T,F)$, with the maximum edge-weight $w$, a
marking $M_0$ is non-live iff there are
	\begin{itemize}
		\item			
$M\wit\in\rset{M_0}$
	(a~witness marking),
\item
			$P\cruc\subseteq P$ (a set of crucial
	places), and
\item
			a nonempty set $T\dead\subseteq T$ (a set of dead
	transitions)
	\end{itemize}			
	such that
	\begin{enumerate}
\itemsep=0.96pt
		\item  $0\leq M\wit(p)\leq w$ for each 	$p\in P\cruc$;
		\item $\rst{N}{(P\cruc,T\smallsetminus T\dead)}$ is
			an IMO net (and is thus  conservative);
		\item all transitions from $T\dead$ are dead
			in $(\rst{N}{(P\cruc,T)},\rst{(M\wit)}{P\cruc})$.
	\end{enumerate}
\end{lemma}

\section{Structural liveness for BIMO nets is in PSPACE}\label{sec:inpspace}

The previous results allow us to give a straightforward proof of the
following theorem; we note that we assume a standard presentation of
BIMO nets, with edge-weights given in binary.

\begin{theorem}\label{thm:BIMOinPSPACE}
	The structural liveness problem (SLP) for BIMO nets is in
	PSPACE (and is thus PSPACE-complete).
\end{theorem}
\begin{proof}
	We suggest a nondeterministic algorithm,
	Algorithm~\ref{ALG:1}; its input consists of a BIMO net $N$ and
	a marking $M_0$ of $N$ (with the values $M_0(p)$ given in
	binary), it works in polynomial space, and it can finish
	successfully if, and only if, $(N,M_0)$ is non-live.
	(This establishes the claim, since NPSPACE$=$PSPACE.)

\begin{algorithm}[ht!]
		\label{ALG:1}\caption{(Nondeterministically) verify non-liveness of a marked BIMO net}
\KwIn{a BIMO net $N=(P,T,F)$, where $P=\{p_1,p_2,\dots,p_m\}$
	and $w$ is the maximum edge-weight; and a marking
	$M_0:P\rightarrow\Nat$.}
	\KwOut{at least one computation returns $true$ if, and only
		if, $(N,M_0)$ is non-live.}
\Begin{
	$\textsc{M} \gets M_0$ $\{\textsc{M}$ is a program variable
		containing the current marking$\}$\;
		$(\textsc{b}_1,\dots,\textsc{b}_m) \gets (false,\dots,false)$
		$\{$each place $p_i$ has an attached boolean
		variable $\textsc{b}_i\}$\;
		\Repeat(){$false$ $\{$\textnormal{hence the cycle
		repeats forever
		if not finishing with a \Return or
a \textbf{fail}$\}$}}{
\For{$i \gets 1$ \textbf{to} $m$}{
	\If{$\textsc{M}(p_i)> 2\cdot w\cdot |P|$}{
		$\textsc{M}(p_i) \gets 2\cdot w\cdot |P|$;
		$\textsc{b}_i \gets true$ $\{\textsc{b}_i$ will not
		change anymore$\}$\;
	}		
}
		\textbf{choose (nondeterministically)}
		$c\in\{1,2,3\}$\;
	\If{$c=1$}{
		\textbf{choose} a transition $t$ that is enabled at
the current marking $M$ stored in $\textsc{M}$\;
$\{$the computation \textbf{fails} if there is no
		such $t\}$\;
		$\textsc{M} \gets M'$ where $M \step{t} M'$\;		
	}
	\If{$c=2$}{
		\textbf{choose} $i\in[1,m]$\;
		\If{$\textsc{M}(p_i)< 2\cdot w\cdot |P|$ and $\textsc{b}_i=true$}{
			$\textsc{M}(p_i) \gets \textsc{M}(p_i)+1$\;
		}
	}
	\If{$c=3$}{
		\textbf{choose} $P\cruc\subseteq P$ such that $\rst{N}{(P\cruc,T)}$
			is an IMO net\;
		\textbf{choose $t\in T$}\;
		\If{$t$ is dead in
		$(\rst{N}{(P\cruc,T)},\rst{\textsc{M}}{P\cruc})$}{
			\Return{true}\;
		}
	}
}
}
\end{algorithm}

\medskip
Inspecting the presented pseudocode,
the fact that Algorithm~\ref{ALG:1} works in polynomial space is
obvious, including
the check at line 25: since $\rst{N}{(P\cruc,T)}$ is an IMO net,
and thus a conservative net, determining whether $\premset{t}$ can be
covered from $\rst{\textsc{M}}{P\cruc}$ is clearly solvable in polynomial
space (when we again recall that PSPACE$=$NPSPACE).

\medskip
	Lemma~\ref{lem:crucialingredgener}
	and Theorem~\ref{thm:upperboundstable}(2)
guarantee that Algorithm~\ref{ALG:1} can return $true$ if, and
	only if, $(N,M_0)$ is non-live:
	\begin{itemize}
		\item	 If	$(N,M_0)$ is non-live, then there are  	$M\wit\in\rset{M_0}$, $P\cruc$, $T\dead$
as given in Lemma~\ref{lem:crucialingredgener}. In this case Algorithm~\ref{ALG:1} can simply perform a
			respective execution $M_0\step{\sigma}M\wit$, by repeatedly choosing the cases $c=1$ and
			$c=2$; 	forgetting the precise marking values above the bound
			$2\cdot w\cdot |P|$ (due to the line $7$) does  not prevent this since there are 	the
			lines $18-19$ in the case $c=2$. Finally the case $c=3$ is chosen, with
the respective $P\cruc$ and some $t\in T\dead$.

		\item If $(N,M_0)$ is live, then all markings stored		in $\textsc{M}$ must be live: due to
     Theorem~\ref{thm:upperboundstable}(2), and the monotonicity of
	net executions (if $M_1\step{\sigma}M_2$, then	$M_1+M_3\step{\sigma}M_2+M_3$), all such
			markings are reachable from live markings (and thus are live themselves).
This fact prevents Algorithm~\ref{ALG:1} from returning $true$.
	\end{itemize}
\end{proof}

\section{Additional remarks}\label{sec:addrem}
As mentioned in the introduction, IO nets model the IO
protocols, hence a subclass of the general population protocols.
The nets modelling the general population protocols, the
\emph{pp-nets} for short, are also
conservative, hence the liveness problem (LP) is also  PSPACE-complete
for them.

On the other hand, in~\cite{JLV22} we elaborate an extension of the
lower-bound proof from~\cite{DBLP:journals/acta/JancarP19} to show
that the structural liveness problem (SLP) is EXPSPACE-hard for
the pp-nets.

\paragraph{Acknowledgements.}
We thank Chana Weil-Kennedy for
useful discussions
with Ji\v{r}{\'{i}} Val\r{u}\v{s}ek. We also thank anonymous reviewers
for their helpful comments.

\end{document}

%% file: figures/example-basic-definitions.tex
\begin{tikzpicture}[node distance=1.3cm,>=stealth',bend angle=30,auto]
    \tikzstyle{place}=[circle,thick,draw=black!75,fill=blue!20,minimum
    size=5mm]
    \tikzstyle{transition}=[rectangle,thick,draw=black!75,
    fill=black!20,minimum size=4mm]
    \tikzstyle{dots}=[circle,thick,draw=none,minimum
    size=6mm]

    \tikzstyle{every label}=[black]

    \begin{scope}
        \node [place,tokens=4] (p1) [label=above:$p_1$] {};
        
        \node [transition] (t1) [label=above:$t_1$,right of=p1] {}
        edge [pre] (p1);
        \node [place] (p2) [label=above:$p_2$,right of=t1] {}
        edge [pre] (t1);

        \node [place] (p3) [label=below:$p_3$,below of=p2] {}
        edge [pre, bend left] (t1)
        edge [post, bend right] (t1);

        \node [transition] (t2) [label=below:$t_2$,right of=p3] {}
        edge [pre] (p2)
        edge [post] (p3);
        
        \node [transition] (t3) [label=below:$t_3$,left of=p3] {}
        edge [pre] (p3);
        
        \node [place] (p4) [label=below:$p_4$,below of=p1] {}
        edge [pre] (t3)
        edge [pre, bend left] (t1)
        edge [post, bend right] (t1);
        
        \node [transition] (t4) [label=below:$t_4$,left of=p4] {}
        edge [pre] (p4)
        edge [post] (p1);
        
        \node [place] (p5) [label=below:$p_5$,below of=p3] {}
        edge [pre] (t2);

        \node [transition] (t5) [label=below:$t_5$,right of=p5] {}
        edge [pre] (p5);
        
        \node [transition] (t6) [label=below:$t_6$,left of=p5] {}
        edge [post] (p5);

        \node [place,tokens=1] (p6) [label=below:$p_6$,left of=t6] {}
        edge [pre, bend left] (t6)
        edge [post, bend right] (t6);
    \end{scope}
\end{tikzpicture}

%% file: figures/hardness1.tex
\begin{tikzpicture}[node distance=1.3cm,>=stealth',bend angle=20,auto]
    \tikzstyle{place}=[circle,thick,draw=black!75,fill=blue!20,minimum
    size=4mm]
    \tikzstyle{transition}=[rectangle,thick,draw=black!75,
    fill=black!20,minimum size=3mm]
    \tikzstyle{dots}=[circle,thick,draw=none,minimum
    size=6mm]

    \tikzstyle{every label}=[black]

    % Cells
    \begin{scope}
        \node [place] (cellai1) [label=left:$p_\lara{i,x}$] {};

        \node [place] (cellbi1) [label=right:$p_\lara{i,x'}$,right=2cm of cellai1] {};
    \end{scope}
    % Head and states
    \begin{scope}[xshift = 0,yshift = 3cm]
        \node [place] (headqi1) [label=left:$p_\lara{q,i}$] {};
        
        \node [place] (headqj1) [label=right:$p_\lara{q',i+m}$,right=2cm of headqi1] {};
    \end{scope}
    % Move
    \begin{scope}[xshift = 1.25cm,yshift = 1.5cm]
        \node [transition] (tmove1) [label=left:$t_\lara{\ins,i}$] {}
        edge [pre] (headqi1)
        edge [post] (headqj1)
        edge [pre] (cellai1)
        edge [post] (cellbi1);
    \end{scope}

    % Cells
    \begin{scope}[xshift = 7cm]
        \node [place] (cellai2) [label=left:$p_\lara{i,x}$] {};

        \node [place] (cellbi2) [label=right:$p_\lara{i,x'}$,right=2cm of cellai2] {};
    \end{scope}
    % Head and states
    \begin{scope}[xshift = 7cm,yshift = 3cm]
        \node [place] (headqi2) [label=left:$p_\lara{q,i}$] {};

        \node [place] (headqj2) [label=right:$p_\lara{q',i+m}$,right=2cm of headqi2] {};
    \end{scope}
    % Move
    \begin{scope}[xshift = 8.25cm,yshift = 2cm]
        \node [place] (move) [label=above:$p_\lara{\ins,i}$] {};
        
        \node [transition] (tmove1) [label=left:$t^{\tbegin}_\lara{\ins,i}$,left of=move] {}
        edge [pre] (headqi2)
        edge [post] (move)
        edge [pre, bend left] (cellai2)
        edge [post, bend right] (cellai2);
        
        \node [transition] (tmove2) [label=right:$t^{\tend}_\lara{\ins,i}$,right of=move] {}
        edge [pre] (move)
        edge [post] (headqj2)
        edge [pre, bend left] (cellbi2)
        edge [post, bend right] (cellbi2);
        
        \node [transition] (tmovehead) [label=below:$t^{\tmove}_\lara{\ins,i}$,below of=move] {} edge [pre] (cellai2) edge [post]
        (cellbi2) edge [pre, bend left] (move) edge
        [post,bend right] (move);
    \end{scope}
\end{tikzpicture}

%% file: figures/hardness2.tex
\begin{tikzpicture}[node distance=1.3cm,>=stealth',bend angle=15,auto]
    \tikzstyle{place}=[circle,thick,draw=black!75,fill=blue!20,minimum
    size=4mm]
    \tikzstyle{transition}=[rectangle,thick,draw=black!75,
    fill=black!20,minimum size=3mm]
    \tikzstyle{dots}=[circle,thick,draw=none,minimum
    size=6mm]

    \tikzstyle{every label}=[black]

    \begin{scope}[yshift = 2.5cm]
        \node [place] (freeold) [label=left:$p\free$] {};

        \node [place] (runold) [label=left:$p\run$,below of=freeold] {};
    \end{scope}
    
    \begin{scope}[xshift = 3cm, yshift = 2.5cm]
        \node [place] (headqa1old) [label=right:$p_\lara{q\acc,1}$] {};

        \node [transition] (taold) [label=below:$t_{\textsc{a}}$,left of=headqa1old] {}
        edge [pre, bend left] (headqa1old)
        edge [post, bend right] (headqa1old)
        edge [pre] (runold)
        edge [post] (freeold);

        \node [transition] (tbold) [label=below:$t_{\textsc{a}}'$,below of=taold] {}
        edge [pre, bend left] (headqa1old)
        edge [post, bend right] (headqa1old)
        edge [pre] (freeold)
        edge [post] (runold);
    \end{scope}

    \begin{scope}[xshift = 8.5cm,yshift = 2cm]
        \node [place] (free) [label=left:$p\free$] {};
    \end{scope}

    \begin{scope}[xshift = 7cm,yshift = 4cm]
        \node [place] (headqi1) [label=left:$p_\lara{q,i}$] {};

        \node [transition] (theadqi1) [label=left:$t_\lara{q,i,q',i'}$,below right of=headqi1] {}
        edge [pre, bend left] (free)
        edge [post, bend right] (free)
        edge [pre] (headqi1);

        \node [transition] (theadqi2) [label=right:$t_\lara{q',i',q,i}$,right of=theadqi1] {}
        edge [pre, bend left] (free)
        edge [post, bend right] (free)
        edge [post] (headqi1);
        
        \node [place] (headqi2) [label=right:$p_\lara{q',i'}$,above right of=theadqi2] {}
        edge [pre] (theadqi1)
        edge [post] (theadqi2);
    \end{scope}

    \begin{scope}[xshift = 7cm]
        \node [place] (cellia) [label=left:$p_\lara{i,a}$] {};

        \node [transition] (tcell1) [label=left:$t_\lara{i,a,i,b}$,above right of=cellia] {}
        edge [pre, bend left] (free)
        edge [post, bend right] (free)
        edge [pre] (cellia);

        \node [transition] (tcell2) [label=right:$t_\lara{i,b,i,a}$,right of=tcell1] {}
        edge [pre, bend left] (free)
        edge [post, bend right] (free)
        edge [post] (cellia);
        
        \node [place] (cellib) [label=right:$p_\lara{i,b}$,below right of=tcell2] {}
        edge [pre] (tcell1)
        edge [post] (tcell2);
    \end{scope}
\end{tikzpicture}

%% file: figures/hardness3.tex
\begin{tikzpicture}[node distance=1.5cm,>=stealth',bend angle=15,auto]
    \tikzstyle{place}=[circle,thick,draw=black!75,fill=blue!20,minimum
    size=4mm]
    \tikzstyle{transition}=[rectangle,thick,draw=black!75,
    fill=black!20,minimum size=3mm]
    \tikzstyle{dots}=[circle,thick,draw=none,minimum
    size=6mm]

    \tikzstyle{every label}=[black]

    \begin{scope}[yshift = 3.5cm]
        \node [place] (free) [label=below:$p\free$] {};

        \node [transition] (tinit1) [label=above:$t_\lara{\init{1}}$,right of=free] {}
        edge [pre] (free);
        
        %\node [transition] (trev1) [label=left:$t_\lara{\trev,1}$,above of=tinit1] {}
        %edge [post] (free);
        
        \node [place] (init1) [label=below:$p_\lara{\init{1}}$,right of=tinit1] {}
        edge [pre] (tinit1);% edge [post] (trev1);

        \node [transition] (tinit2) [label=above:$t_\lara{\init{2}}$,right of=init1] {}
        edge [pre] (init1);
        
        %\node [transition] (trev2) [label=right:$t_\lara{\trev,2}$,above of=tinit2] {}
        %edge [post] (free);
        
        \node [place] (init2) [label=below:$p_\lara{\init{2}}$,right of=tinit2] {}
        edge [pre] (tinit2); %edge [post] (trev2);
        
        \node [dots] (initd1) [right of=init2] {\dots};

        \node [transition] (tinitn) [label=above:$t_\lara{\init{n}}$,right of=initd1] {};
        
        \node [place] (initn) [label=below:$p_\lara{\init{n}}$,right of=tinitn] {}
        edge [pre] (tinitn);
        
        \node [transition] (trun) [label=above:$t\run$,right of=initn] {}
        edge [pre] (initn);
        
        \node [place] (run) [label=above:$p\run$,right of=trun] {}
        edge [pre] (trun);
    \end{scope}
    \begin{scope}[xshift = 9cm]
        \node [place] (headq01) [label=below:$p_\lara{q_0,1}$] {}
        edge [pre, bend left] (trun)
        edge [post, bend right] (trun);

%            \node [dots] (headd1) [right=0.5cm of headq01] {\dots};

        \node [place] (headqa1)
    [label=below:$p_\lara{q\acc,1}$,right=2cm of
    headq01] {};

        \node [transition] (ta) [label=right:$t_{\textsc{a}}$,above of=headqa1] {}
        edge [pre, bend left] (headqa1)
        edge [post, bend right] (headqa1)
        edge [pre, bend right] (run)
        edge [post, bend left] (free);
    \end{scope}
    \begin{scope}[xshift = 1.5cm]
        \node [place] (cell1a1) [label=below:$p_\lara{1,x_1}$] {}
        edge [pre, bend left] (tinit1)
        edge [post, bend right] (tinit1);

        \node [place] (cell2a2) [label=below:$p_\lara{2,x_2}$,right of=cell1a1] {}
        edge [pre, bend left] (tinit2)
        edge [post, bend right] (tinit2);

        \node [dots] (celld1) [right=0.5cm of cell2a2] {\dots};

        \node [place] (cellnan) [label=below:$p_\lara{n,x_n}$,right=0.5cm of celld1] {}
        edge [pre, bend left] (tinitn)
        edge [post, bend right] (tinitn);
    \end{scope}
    \begin{scope}[xshift = 1.5cm,yshift = 5.5cm]
        \node [transition] (trev1) [label=left:$t_\lara{\trev,1}$] {}
        edge [post,bend right] (free)
        edge [pre,bend left] (init1);
    \end{scope}
    \begin{scope}[xshift = 4.5cm,yshift = 5.5cm]
        \node [transition] (trev1) [label=right:$t_\lara{\trev,2}$] {}
        edge [post,bend right] (free)
        edge [pre,bend left] (init2);
    \end{scope}
    \begin{scope}[xshift = 9cm,yshift = 5.5cm]
        \node [transition] (trevn) [label=right:$t_\lara{\trev,n}$] {}
        edge [pre,bend left] (initn);

        \node [dots] (initd2) [left of=trevn] {\dots};
    \end{scope}
\end{tikzpicture}

%% file: figures/example1.tex
\begin{tikzpicture}[node distance=1.7cm,>=stealth',bend angle=15,auto]
    \tikzstyle{place}=[circle,thick,draw=black!75,fill=blue!20,minimum
    size=5mm]
    \tikzstyle{transition}=[rectangle,thick,draw=black!75,
    fill=black!20,minimum size=4mm]
    \tikzstyle{dots}=[circle,thick,draw=none,minimum
    size=6mm]

    \tikzstyle{every label}=[black]

    \begin{scope}
        \node [place] (pB1) [label=above:$p_1$] {};

        \node [transition] (tB1) [label=below:$t_1$,below right of=pB1] {}
        edge [pre, very thick] (pB1);
        
        \node [transition] (tB2) [label=above left:$t_2$,right of=pB1] {}
        edge [post, very thick] (pB1);
        
        \node [place,tokens=1] (pB2) [label=below:$p_2$,right of=tB2] {}
        edge [post, very thick] (tB2)
        edge [pre, bend left, densely dotted, thick] (tB1)
        edge [post, bend right, densely dotted, thick] (tB1);
        
        \node [transition] (tB3) [label=below:$t_3$,above right of=pB2] {}
        edge [post, very thick] (pB2);
        
        \node [place,tokens=1] (pB3) [label=above:$p_3$,above left of=tB3] {}
        edge [post, thick] (tB3)
        edge [pre, bend left, densely dotted, thick] (tB2)
        edge [post, bend right, densely dotted, thick] (tB2);

        \node [transition] (tB4) [label=above:$t_4$,left of=pB3] {}
        edge [post, very thick] (pB3)
        edge [pre, bend left, densely dotted, thick] (pB2)
        edge [post, bend right, densely dotted, thick] (pB2);

        \node [place] (pB4) [label=above:$p_4$,below left of=tB4] {}
        edge [post, very thick] (tB4)
        edge [pre, bend left, densely dotted, thick] (tB3)
        edge [post, bend right, densely dotted, thick] (tB3);
        
        \node [transition] (tB5) [label=above:$t_5$,left of=pB4] {}
        edge [post, very thick] (pB4)
        edge [pre, very thick] (pB1);

        \node [place] (pA1) [label=below:$p_5$,below right of=pB2] {}
        edge [pre, very thick] (tB1);

        \node [transition] (tA1) [label=below:$t_6$,right of=pA1] {}
        edge [pre, very thick] (pA1);

        \node [place,tokens=4] (pA2) [label=above:$p_6$,above right of=tA1] {}
        edge [pre, very thick] (tA1);

        \node [transition] (tA2) [label=above:$t_7$,above of=tA1] {}
        edge [post, bend left, densely dotted, thick] (pB2)
        edge [pre, bend right, densely dotted, thick] (pB2)
        edge [pre, bend left, densely dotted, thick] (pA1)
        edge [post, bend right, densely dotted, thick] (pA1)
        edge [post, very thick] (pB3)
        edge [pre, very thick] (pA2);

        \node [place,tokens=4] (pA3) [label=above:$p_7$,right of=tA1] {}
        edge [pre, very thick] (tA1);
        \node [transition] (tA3) [label=above:$t_8$,above right of=pA3] {}
        edge [pre, very thick] (pA2)
        edge [post, very thick] (pA3);
        \node [transition] (tA4) [label=above:$t_9$,right of=pA3] {}
        edge [pre, very thick] (pA3)
        edge [post, very thick] (pA2);
    \end{scope}
\end{tikzpicture}

%% file: figures/example2.tex
\begin{tikzpicture}[node distance=1.6cm,>=stealth',bend angle=10,auto]
    \tikzstyle{place}=[circle,thick,draw=black!75,fill=blue!20,minimum
    size=5mm]
    \tikzstyle{transition}=[rectangle,thick,draw=black!75,
    fill=black!20,minimum size=3mm]
    \tikzstyle{dots}=[circle,thick,draw=none,minimum
    size=6mm]

    \tikzstyle{every label}=[black]

    \begin{scope}
        \node [place] (p1) [label=left:$p_1$] {};
        \node [transition] (t1) [label=below:$t_1$,right of=p1] {}
        edge [pre, very thick] (p1);
        \node [place] (p2) [label=below:$p_2$,right of=t1] {}
        edge [pre, very thick] (t1);
        \node [transition] (t2) [label=below:$t_2$,right of=p2] {}
        edge [pre, very thick] (p2);
        \node [transition] (t3) [label=left:$t_3$,below of=p2] {}
        edge [post, very thick] (p1);
        \node [place] (p3) [label=right:$p_3$,right of=t2] {}
        edge [pre, very thick] (t2)
        edge [post, very thick] (t3);
        
        \node [transition] (t4) [label=above:$t_4$,above of=p2] {}
        edge [pre, very thick] (p1)
        edge [pre, bend left, densely dotted, thick] (p2)
        edge [post, bend right, densely dotted, thick] (p2)
        edge [post, bend left, densely dotted, thick] (p3)
        edge [pre, bend right, densely dotted, thick] (p3);
        \node [place] (p4) [label=above:$p_4$,left of=t4] {}
        edge [pre, very thick] (t4);
        \node [transition] (t5) [label=above:$t_5$,left of=p4] {}
        edge [pre, bend left, densely dotted, thick] (p4)
        edge [post, bend right, densely dotted, thick] (p4)
        edge [pre, very thick] (p1);

        \node [place] (p5) [label=left:$p_5$,below of=t3] {};
        \node [transition] (t6) [label=left:$t_6$,left of=t3] {}
        edge [pre, very thick] (p1)
        edge [post, very thick] (p5)
        edge [pre, bend left, densely dotted, thick] (p2)
        edge [post, bend right, densely dotted, thick] (p2);
        \node [transition] (t7) [label=right:$t_7$,right of=t3] {}
        edge [post, very thick] (p3)
        edge [pre, very thick, bend left] (p5)
        edge [post, very thick, bend right] (p5);
    \end{scope}
\end{tikzpicture}

%% file: figures/example3.tex
\begin{tikzpicture}[node distance=1.7cm,>=stealth',bend angle=25,auto]
    \tikzstyle{place}=[circle,thick,draw=black!75,fill=blue!20,minimum
    size=4mm]
    \tikzstyle{transition}=[rectangle,thick,draw=black!75,
    fill=black!20,minimum size=3mm]
    \tikzstyle{dots}=[circle,thick,draw=none,minimum
    size=6mm] \tikzstyle{every label}=[black]

    \begin{scope}
        \node [place,tokens=1] (p1) [label=above:$p_1$] {};
        
        \node [transition] (t1) [label=below right:$t_1$,right of=p1] {}
        edge [pre,very thick] (p1);
        
        \node [place] (p2) [label=above:$p_2$,right of=t1] {}
        edge [pre,very thick] (t1);
        
        \node [transition] (t2) [label=right:$t_2$,above right of=p2] {}
        edge [pre,very thick] (p2)
        edge [post,very thick] (p1);
        
        \node [place,tokens=1] (p3) [label=above right:$p_3$,below left of=p1] {};

        \node [transition] (t3) [label=below:$t_3$,below left of=p3] {}
        edge [pre, bend left, densely dotted, thick] (p1)
        edge [post, bend right, densely dotted, thick] (p1)
        edge [pre,very thick] (p3);

        \node [transition] (t4) [label=above:$t_4$,above left of=p3] {}
        edge [post, bend left, densely dotted, thick] (p1)
        edge [pre, bend right, densely dotted, thick] (p1)
        edge [post,very thick] (p3);

        \node [place] (p4) [label=below:$p_4$,below right of=t3] {}
        edge [pre, bend left, densely dotted, thick] (t1)
        edge [post, bend right, densely dotted, thick] (t1)
        edge [pre,very thick] (t3)
        edge [post,very thick] (t4);

        \node [transition] (t5) [label=above:$t_5$,below right of=p2] {}
        edge [pre, bend left, densely dotted, thick] (p2)
        edge [post, bend right, densely dotted, thick] (p2)
        edge [pre,very thick] (p3);

        \node [place] (p5) [label=below:$p_5$,right of=t5] {}
        edge [pre, bend left, densely dotted, thick] (t2)
        edge [post, bend right, densely dotted, thick] (t2)
        edge [pre,very thick] (t5);
        
        \node [transition] (t6) [label=below right:$t_6$,below left of=t5] {}
        edge [pre, bend left, densely dotted, thick] (p1)
        edge [post, bend right, densely dotted, thick] (p1)
        edge [pre,very thick] (p5);

        \node [place,tokens=1] (p6) [label=below:$p_6$,below of=t6] {}
        edge [pre,very thick] (t6);
        
        \node [transition] (t7) [label=below:$t_7$,above left of=p6] {}
        edge [pre, bend left, densely dotted, thick] (p2)
        edge [post, bend right, densely dotted, thick] (p2)
        edge [pre,very thick] (p6)
        edge [post,very thick] (p3);
    \end{scope}
\end{tikzpicture}

%% file: figures/example-bio.tex
\begin{tikzpicture}[node distance=1.3cm,>=stealth',bend angle=15,auto]
    \tikzstyle{place}=[circle,thick,draw=black!75,fill=blue!20,minimum
    size=4mm]
    \tikzstyle{transition}=[rectangle,thick,draw=black!75,
    fill=black!20,minimum size=3mm]
    \tikzstyle{dots}=[circle,thick,draw=none,minimum
    size=6mm]

    \tikzstyle{every label}=[black]

    \begin{scope}[xshift = 7cm]
        \node [place] (p1) [label=above:$p_1$] {};
        \node [transition] (t1) [label=right:$t_1$,below of=p1] {}
        edge [pre] (p1);
        \node [place,tokens=1] (p3) [label=above right:$p_3$,below of=t1] {}
        edge [pre] (t1);

        \node [transition] (t16) [label=above:$t_{16}$,left=0.5cm of p3] {}
        edge [pre] (p3);
        \node [transition] (t15) [label=above:$t_{15}$,above=0.5cm of t16] {}
        edge [pre] (p3);

        \node [transition] (t9) [label=above:$t_9$,above right=0.5cm of p1] {}
        edge [pre] (p1);
        \node [transition] (t10) [label=above:$t_{10}$,right=0.5cm of t9] {}
        edge [post] (p1);
        \node [place] (p2) [label=left:$p_2$,below right=0.5cm of t10] {}
        edge [pre] (t9)
        edge [post] (t10)
        edge [dashed] (t1);
        \node [transition] (t13) [label=above:$t_{13}$,left=0.5cm of p1] {}
        edge [pre] (p1);
        \node [transition] (t14) [label=above:$t_{14}$,below=0.5cm of t13] {}
        edge [pre] (p1);
    \end{scope}
    \begin{scope}[xshift = 2cm]
        \node [place] (p4) [label=above:$p_4$] {};
        \node [transition] (t2) [label=left:$t_2$,below of=p4] {}
        edge [pre] (p4);
        \node [place] (p6) [label=above left:$p_6$,below of=t2] {}
        edge [pre] (t2);

        \node [transition] (t20) [label=above:$t_{20}$,right=0.5cm of p6] {}
        edge [pre] (p6);
        \node [transition] (t19) [label=above:$t_{19}$,above=0.5cm of t20] {}
        edge [pre] (p6);

        \node [transition] (t11) [label=above:$t_{11}$,above left=0.5cm of p4] {}
        edge [pre] (p4);
        \node [transition] (t12) [label=above:$t_{12}$,left=0.5cm of t11] {}
        edge [post] (p4);
        \node [place,tokens=2] (p5) [label=right:$p_5$,below left=0.5cm of t12] {}
        edge [pre] (t11)
        edge [post] (t12)
        edge [dashed] (t2);
        \node [transition] (t17) [label=above:$t_{17}$,right=0.5cm of p4] {}
        edge [pre] (p4);
        \node [transition] (t18) [label=above:$t_{18}$,below=0.5cm of t17] {}
        edge [pre] (p4);
    \end{scope}
    \begin{scope}[yshift = 3cm]
        \node [transition] (t3) [label=above:$t_7$] {}
        edge [post] (p5);
        \node [place] (p8) [label=above:$p_{12}$,right=0.5cm of t3] {}
        edge [pre] (t3)
        edge [dashed] (t14)
        edge [dashed] (t16);
        \node [transition] (t4) [label=above:$t_8$,right=0.5cm of p8] {}
        edge [pre] (p8)
        edge [post,bend right] (p5);
    \end{scope}
    \begin{scope}[xshift = 5cm,yshift = 3cm]
        \node [place] (p9) [label=above:$p_7$] {}
        edge [pre] (t4)
        edge [dashed] (t9)
        edge [dashed] (t10);
        \node [transition] (t5) [label=above:$t_3$,right=0.5cm of p9] {}
        edge [pre] (p9);
        \node [place,tokens=1] (p10) [label=above:$p_8$,right=0.5cm of t5] {}
        edge [pre] (t5)
        edge [dashed] (t17)
        edge [dashed] (t19);
        \node [transition] (t6) [label=above:$t_4$,right=0.5cm of p10] {}
        edge [pre] (p10)
        edge [post,bend left] (p2);
        \node [place] (p11) [label=above:$p_9$,right=0.5cm of t6] {}
        edge [pre] (t6)
        edge [dashed] (t18)
        edge [dashed] (t20);
        \node [transition] (t7) [label=above:$t_5$,right=0.5cm of p11] {}
        edge [pre] (p11)
        edge [post] (p2);
    \end{scope}
    \begin{scope}[xshift = 5cm,yshift = 5cm]
        \node [place] (p12) [label=above:$p_{10}$] {}
        edge [dashed] (t11)
        edge [dashed] (t12)
        edge [pre,bend left] (t7);
        \node [transition] (t8) [label=above:$t_6$,left=0.5cm of p12] {}
        edge [pre] (p12);
        \node [place] (p7) [label=above:$p_{11}$,left=0.5cm of t8] {}
        edge [dashed] (t13)
        edge [dashed] (t15)
        edge [post,bend right] (t3)
        edge [pre] (t8);
    \end{scope}

    \draw[densely dotted, thick] (t7) to[in=90,out=280] ($(10cm,0cm)$)
         to[in=0,out=270] (p3);
    \draw[densely dotted, thick] (t6) to[in=90,out=270] ($(9.5cm,0cm)$)
         to[in=10,out=270] (p3);
    \draw[densely dotted, thick] (t5) to[in=90,out=270] ($(5cm,-2cm)$)
         to[in=225,out=270] (p3);
    \draw[densely dotted, thick] (t4) to[in=90,out=270] ($(4cm,0cm)$)
         to[in=90,out=270] ($(4cm,-2cm)$)
         to[in=315,out=270] (p6);
    \draw[densely dotted, thick] (t3) to[in=90,out=260] ($(-1cm,0cm)$)
         to[in=170,out=270] (p6);
    \draw[densely dotted, thick] (t8) to[in=90,out=270] ($(-1cm,3cm)$)
         to[in=90,out=270] ($(-1.5cm,0cm)$)
         to[in=180,out=270] (p6);
\end{tikzpicture}

%% file: figures/example-extension.tex
\begin{tikzpicture}[node distance=1.3cm,>=stealth',bend angle=30,auto]
    \tikzstyle{place}=[circle,thick,draw=black!75,fill=blue!20,minimum
    size=4mm]
    \tikzstyle{transition}=[rectangle,thick,draw=black!75,
    fill=black!20,minimum size=3mm]
    \tikzstyle{dots}=[circle,thick,draw=none,minimum
    size=6mm]

    \tikzstyle{every label}=[black]

    \begin{scope}
        \node [place] (p) [label=below:$p_1$]
        {};
        
        \node [transition] (t)
        [label=below:$t$,right of=p] {}
        edge [post,bend right] node[swap]{3}  (p)
        edge [pre,bend left] node {2} (p);
        
        \node [place] (pd1) [label=below:$p_2$,right of=t] {}
        edge [pre] (t);
    \end{scope}
    \begin{scope}[xshift = 8cm,yshift = 0cm]
        \node [place] (p) [label=below:$p_\lara{1,2}$]
        {};

        \node [place] (p2) [label=above:$p_\lara{1,1}$,above of=p] {};
        \node [transition] (t11) [label=left:$t_\lara{1,1}$,left of=p2] {}
        edge [post] (p)
        edge [pre] (p2);
        \node [transition] (t12) [label=left:$t_\lara{1,2}$,left of=p] {}
        edge [pre] (p);
        \node [place] (p3) [label=below:$p_\lara{1,3}$,below of=p] {}
        edge [pre] (t12);
        \node [transition] (t13) [label=left:$t_\lara{1,3}$,left of=p3] {}
        edge [post] (p2)
        edge [pre] (p3);
        
        \node [transition] (t)
        [label=below:$t'$,right of=p] {}
        edge [post,bend right] (p)
        edge [pre,bend left] (p)
        edge [pre,bend right] (p2)
        edge [post,bend left] (p2)
        edge [post] (p3);
        
        \node [place] (pd1) [label=below:$p_{\lara{2,1}}$,right of=t] {}
        edge [pre] (t);
    \end{scope}
\end{tikzpicture}